\documentclass{CSML}
\pdfoutput=1

\usepackage{lastpage}
\newcommand{\dOi}{14(1:24)2018}
\lmcsheading{}{1--\pageref{LastPage}}{}{}%
{Feb.~28, 2017}{Mar.~29, 2018}{}

\usepackage{amsmath}
\usepackage{amsfonts}
\usepackage{amssymb}
\usepackage{amsthm}
\usepackage{mathrsfs}  
\usepackage{hhline}
\usepackage[]{float}
\usetikzlibrary{automata,positioning}
\usetikzlibrary{matrix}
\usepackage{epic,eepic}
\usepackage{stmaryrd}
\usepackage[normalem]{ulem}
\usepackage{pstricks}

\usepackage{array}
\usepackage{multicol}
\usepackage{tabularx}
\usepackage{mathtools}  

\usepackage{hyperref}
\hypersetup{hidelinks}

\newenvironment{conditions}
{%
   \begin{list}{\rm (\theenumi)}%
   {\noindent%
     \usecounter{enumi}%
     \setlength{\topsep}{2pt}%
     \setlength{\partopsep}{0pt}%
				\setlength{\itemsep}{2pt}%
     \setlength{\parsep}{0pt}%
     \setlength{\leftmargin}{2.5em}%
     \setlength{\labelwidth}{1.5em}%
     \setlength{\labelsep}{0.5em}%
     \setlength{\listparindent}{0pt}%
     \setlength{\itemindent}{0pt}%
   }%
}%
{\end{list}}%

\newenvironment{conditionsJE}[2][]
{\begin{enumerate}[
	series=tests,
	label= $(\text{#2}_{\arabic*})$,
	topsep=2pt,
	partopsep=0pt,
	itemsep=2pt,
	parsep=0pt,
	leftmargin=3em,
	labelwidth=1.5em,
	labelsep=0.5em,
	listparindent=0pt,
	itemindent=0pt,
	#1]
}{\end{enumerate}}

\newcommand{\cdp}{\mkern1mu\cdotp\!}

\newcommand{\clos}[1]{\overline{#1}}
\newcommand{\closI}[1]{{\overline{#1}\,}^1}
\newcommand{\closII}[1]{{\overline{#1}\,}^2}
\newcommand{\closIII}[1]{{\overline{#1}\,}^3}

\newcommand{\Ult}{\operatorname{Ult}}
\newcommand{\Stab}{\operatorname{Stab}}

\newcommand{\shuffle}{\mathrel{\llcorner\!\llcorner\!\!\!\lrcorner}}

\def\dif{\mathrel{\triangle}}

\renewcommand{\geq}{\geqslant}
\renewcommand{\leq}{\leqslant}

\newcommand{\bG}{\mathbf{G}}

\newcommand\sB{\mathscr{B}}

\newcommand{\cA}{\mathcal{A}}
\newcommand{\cB}{\mathcal{B}}

\newcommand{\cF}{\mathcal{F}}
\newcommand{\cG}{\mathcal{G}}
\newcommand{\cJ}{\mathcal{J}}
\newcommand{\cL}{\mathcal{L}}
\newcommand{\cP}{\mathcal{P}}
\newcommand{\cS}{\mathcal{S}}

\newcommand{\tvi}{\vrule height 12pt depth 5 pt width 0 pt}
\newcommand{\ptvi}{\vrule height 10pt depth 4 pt width 0 pt}

\newcommand{\calJ}{\mathrel{\mathcal{J}}}

\newcommand{\leJ}{\leqslant_\mathcal{J}}

\newcommand{\BPolL}{\sB\text{Pol}\,\,\cL}
\newcommand{\BPolG}{\sB\text{Pol}\,\,\cG}

\newcommand{\PolG}{\text{Pol}\,\,\cG}
\newcommand{\coPolG}{\text{co-Pol}\,\,\cG}
\newcommand{\PolL}{\text{Pol}\,\,\cL}

\newcommand{\coPolL}{\text{co-Pol}\,\,\cL}

\newcommand\nstarr[2]{\hfil\raisebox{#1pt}{$#2$}\hfil}
\newcommand\nstar[1]{\nstarr{2}{#1}}

\newcommand\wstarr[2]{\setbox3=\hbox{$*$}%
\hspace{2pt}\raisebox{8pt}{\smash{$*$}}\hspace{-2pt}\hspace{-\wd3}\nstarr{#1}{#2}}
\newcommand\wstar[1]{\wstarr{2}{#1}}

\begin{document}

\title{A survey on difference hierarchies of regular languages}

\author[O. Carton]{Olivier Carton}
\address{IRIF, CNRS and Universit\'e Paris-Diderot}
\email{olivier.carton@irif.fr}

\author[D. Perrin]{Dominique Perrin}
\address{Laboratoire d'informatique Gaspard-Monge, Universit\'e de Marne-la-Vall\'ee}
\email{dominique.perrin@esiee.fr}

\author[J.-\'E. Pin]{Jean-\'Eric Pin}
\address{IRIF, CNRS and Universit\'e Paris-Diderot} \email{jean-eric.pin@irif.fr} \thanks{The
third author is partially funded from the European Research Council (ERC) under the
European Union's Horizon 2020 research and innovation programme (grant agreement No
670624). The first and third authors are partially funded by the DeLTA project
(ANR-16-CE40-0007)}

\subjclass{Formal languages and automata theory, Regular languages, Algebraic
language theory}

\begin{abstract}
	Difference hierarchies were originally introduced by Hausdorff and they play an
	important role in descriptive set theory. In this survey paper, we study difference
	hierarchies of regular languages. The first sections describe standard techniques
	on difference hierarchies, mostly due to Hausdorff. We illustrate these techniques
	by giving decidability results on the difference hierarchies based on shuffle
	ideals, strongly cyclic regular languages and the polynomial closure of group
	languages.
\end{abstract}

\maketitle

\begin{center}
	\emph{Dedicated to the memory of Zolt\'an \'Esik.}
\end{center}


\section{Introduction}\label{sec:Introduction}

Consider a set $E$ and a set $\cF$ of subsets of $E$ containing the empty set. The
general pattern of a difference hierarchy is better explained in a picture. Saturn's
rings-style Figure \ref{fig:Saturn} represents a decreasing sequence
\[
  X_1 \supseteq X_2\supseteq X_3 \supseteq X_{4} \supseteq X_5
\]
of elements of $\cF$. The grey part of the picture corresponds to the set $(X_1 -
X_2) + (X_3 - X_4) + X_5$, a typical element of the fifth level of the difference
hierarchy defined by $\cF$. Similarly, the $n$-th level of the difference hierarchy
defined by $\cF$ is obtained by considering length-$n$ decreasing nested sequences of
sets.
\begin{figure}[H]
\begin{center}
\includegraphics{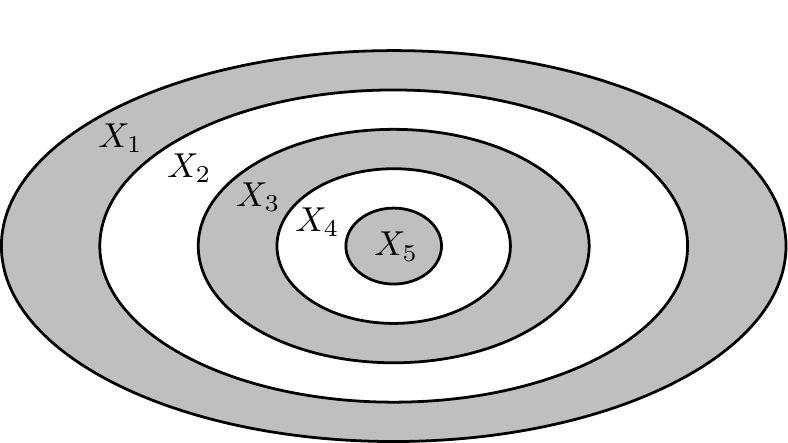}
\end{center}
\caption{Five subsets of $E$.}\label{fig:Saturn}
\end{figure}
\noindent Difference hierarchies were originally introduced by Hausdorff
\cite{Hausdorff14, Hausdorff49, Hausdorff57}. They play an important role in
descriptive set theory \cite[Section 11]{Wadge12} and also yield a hierarchy on
complexity classes known as the Boolean hierarchy \cite[Section
3]{KoblerSchoningWagner87}, \cite[Section 2]{Wechsung85}, \cite{CaiHemachandra86},
\cite[Section 3]{CaiHartmanisetal88}. Difference hierarchies were also used in the
study of $\omega$-regular languages \cite{Carton96, CartonPerrin94, CartonPerrin97,
CartonPerrin97b, CartonPerrin99, Wagner79}.

The aim of this paper is to survey difference hierarchies of regular languages.
Decidability questions for difference hierarchies over regular languages were studied
in \cite{GlasserSchmitz01} and more recently by Glasser, Schmitz and Selivanov in
\cite{GlasserSchmitzSelivanov16}. The latter article is the reference paper on this
topic and contains an extensive bibliography, to which we refer the interested
reader. However, paper \cite{GlasserSchmitzSelivanov16} focuses on forbidden patterns
in automata, a rather different perspective than ours.

We first present some general results on difference hierarchies and their connection
with closure operators. The results on approximation of Section
\ref{sec:Approximation}, first presented in \cite{Carton97}, lead in some cases to
convenient algorithms to compute chain hierarchies.

Next we turn to algebraic methods. Indeed, a great deal of results on regular
languages are obtained through an algebraic approach. Typically, combinatorial
properties of regular languages --- being star-free, piecewise testable, locally
testable, etc. --- translate directly to algebraic properties of the syntactic monoid
of the language (see \cite{Pin95a} for a survey). It is therefore natural to expect a
similar algebraic approach when dealing with difference hierarchies. However, things
are not that simple. First, one needs to work with \emph{ordered} monoids, which are
more appropriate for classes of regular languages not closed under complement.
Secondly, although Proposition \ref{prop:BnV} yields a purely algebraic characterization
of the difference hierarchy, it does not lead to decidability results, except for
some special cases. Two such cases are presented at the end of the paper. The first
one studies the difference hierarchy of the polynomial closure of a lattice of
regular languages. The main result, Corollary \ref{cor:Bn(PolG)}, which appears to be
new, states that the difference hierarchy induced by the polynomial of group
languages is decidable. The second case, taken from \cite{Carton97}, deals with
cyclic and strongly cyclic regular languages.

Our paper is organised as follows. Prerequities are presented in Section
\ref{sec:Prerequisites}. Section \ref{sec:Difference hierarchies} covers the results
of Hausdorff on difference hierarchies and Section \ref{sec:Closure operators} is a
brief summary on closure operators. The results on approximation form the core of
Section \ref{sec:Approximation}. Decidability questions on regular languages are
introduced in Section \ref{sec:Decidabilty questions}. Section \ref{sec:Chains and
difference hierarchies} on chains is inspired by results of descriptive set theory.
Two results that are not addressed in \cite{GlasserSchmitzSelivanov16} are presented
in Sections \ref{sec:polynomial closure} and \ref{sec:Cyclic}. The final Section 
\ref{sec:Conclusion} opens up some perspectives.


\section{Prerequisites}\label{sec:Prerequisites}

In this section, we briefly recall the following notions: upper sets, ordered monoids, 
stamps and syntactic objects.

Let $E$ be a preordered set. An \emph{upper set} of $E$ is a subset $U$ of $E$ such
that the conditions $x \in U$ and $x \leq y$ imply $y \in U$. An \emph{ordered
monoid} is a monoid $M$ equipped with a partial order $\leq$ compatible with the
product on $M$: for all $x, y, z \in M$, if $x \leq y$ then $zx \leq zy$ and $xz \leq
yz$.

A \emph{stamp} is a surjective monoid morphism $\varphi:A^* \to M$ from a finitely
generated free monoid $A^*$ onto a finite monoid $M$. If $M$ is an ordered
monoid, $\varphi$ is called an \emph{ordered stamp}. 

The \emph{restricted direct product} of two stamps $\varphi_1:A^* \to M_1$ and
$\varphi_2:A^* \to M_2$ is the stamp $\varphi$ with domain $A^*$ defined by
$\varphi(a) = (\varphi_1(a), \varphi_2(a))$. The image of $\varphi$ is an [ordered]
submonoid of the [ordered] monoid $M_1 \times M_2$.
\begin{center}
  \includegraphics{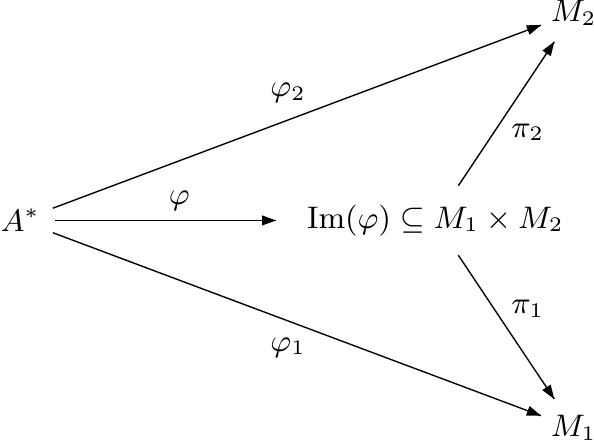}
\end{center}
Stamps and ordered stamps are used to recognise languages. A language $L$ of $A^*$ is
\emph{recognised by a stamp $\varphi: A^* \to M$} if there exists a subset $P$ of $M$
such that $L = \varphi^{-1}(P)$. It is \emph{recognised by an ordered stamp $\varphi:
A^* \to M$} if there exists an upper set $U$ of $M$ such that $L = \varphi^{-1}(U)$.

The syntactic preorder of a language was first introduced by Sch\"utzenberger in
\cite[p. 10]{Schutzenberger56}. Let $L$ be a language of $A^*$. The \emph{syntactic
preorder} of $L$ is the relation $\leq_L$ defined on $A^*$ by $u \leq_L v$ if and
only if, for every $x, y \in A^*$,
\[
	xuy \in L \implies xvy \in L.
\]
The associated equivalence relation $\sim_L$, defined by $u \sim_L v$ if $u \leq_L v$
and $v\leq_L u$, is the \emph{syntactic congruence} of $L$ and the quotient monoid
$M(L) = A^*/{\sim_L}$ is the \emph{syntactic monoid} of $L$. The natural morphism
$\eta: A^* \to A^*/{\sim_L}$ is the \emph{syntactic stamp} of $L$. The
\emph{syntactic image} of $L$ is the set $P = \eta(L)$.

The \emph{syntactic order} $\leq_P$ is defined on $M(L)$ as follows: $u \leq_P v$ if and
only if for all $x,y \in M(L)$,
$$
	xuy \in P \implies xvy \in P
$$
The partial order $\leq_P$ is stable and the resulting ordered monoid $(M(L), \leq_P)$
is called the \emph{syntactic ordered monoid} of $L$. Note that $P$ is now an upper set 
of $(M(L), \leq_P)$ and $\eta$ becomes an ordered stamp, called the syntactic ordered 
stamp of $L$.


\section{Difference hierarchies}\label{sec:Difference hierarchies}

Let $E$ be a set. In this article, a \emph{lattice} is simply a collection of subsets
of $E$ containing $\emptyset$ and $E$ and closed under taking finite unions and
finite intersections. A \emph{lattice} closed under complement is a \emph{Boolean
algebra}. Throughout this paper, we adopt Hausdorff's convention to denote union 
additively, set difference by a minus sign and intersection as a product. We also 
sometimes denote $L^c$ the complement of a subset $L$ of a set $E$. 

Let $\cF$ be a set of subsets of $E$ containing the empty set. We set $\cB_0(\cF) =
\{\emptyset\}$ and, for each integer $n\geq 1$, we let $\cB_n(\cF)$ denote the class
of all sets of the form
\begin{equation}
	  X = X_1 - X_2 + {} \dotsm {} \pm X_n  \label{eq:Dn}
\end{equation}
where the sets $X_i$ are in $\cF$ and satisfy $X_1 \supseteq X_2\supseteq X_3
\supseteq {}\dotsm{} \supseteq X_n$. By convention, the expression on the right hand
side of (\ref{eq:Dn}) should be evaluated from left to right, but given the
conditions on the $X_i$'s, it can also be evaluated as
\begin{equation}
	(X_1-X_2) + (X_3-X_4) + (X_5-X_6) + \dotsm \label{eq:Dn with ()}
\end{equation}
Since the empty set belongs to $\cF$, one has $\cB_n(\cF) \subseteq \cB_{n+1}(\cF)$
for all $n \geq 0$ and the classes $\cB_n(\cF)$ define a hierarchy within the Boolean
closure of $\cF$. Moreover, the following result, due to
Hausdorff~\cite{Hausdorff49}, holds:

\begin{thm}\label{thm:Hausdorff}
	Let $\cF$ be a lattice of subsets of $E$. The union of the classes $\cB_n(\cF)$ for
	$n\geq 0$ is the Boolean closure of $\cF$.
\end{thm}

\begin{proof} Let $\cB(\cF) = \cup_{n\geq 1}\cB_n(\cF)$. By construction, every
element of $\cB_n(\cF)$ is a Boolean combination of members of $\cF$ and thus
$\cB(\cF)$ is contained in the Boolean closure of $\cF$. Moreover $\cB_1(\cF) = \cF$
and thus $\cF\subseteq \cB(\cF)$. It is therefore enough to prove that $\cB(\cF)$ is
closed under complement and finite intersection. If $X = X_1-X_2+{} \dotsm {}\pm
X_n$, one has
$$
	E-X = E-X_1+X_2- {}\dotsm {} \mp X_n
$$
and thus $X\in \cB(\cF)$ implies $E-X\in \cB(\cF)$. Thus $\cB(\cF)$ is closed under
complement.

Let $X = X_1 - X_2 + {} \dotsm {} \pm X_n$ and $Y = Y_1 - Y_2 + {} \dotsm {}\pm Y_m$
be two elements of $\cB(\cF)$. Let
\begin{align*}
  Z &= Z_1 - Z_2 + {} \dotsm {} \pm Z_{n+m-1} \\
\noalign{\noindent with}
Z_k &= \sum_{\substack{i+j = k+1\\ \text{$i$ and $j$ not both even}}} X_iY_j \\
\noalign{\noindent Therefore}
	Z_1 &= X_1Y_1,\\
	Z_2 &= X_1Y_2 + X_2Y_1,\\
	Z_3 &= X_1Y_3 + X_3Y_1,\\
	Z_4 &= X_1Y_4 + X_2Y_3 + X_3Y_2 + X_4Y_1,\\
			&\vdotswithin{=} \\
	Z_{n+m-1}&=\begin{cases}
								X_nY_m   &\text{if $m$ and $n$ are not both even}\\
								\emptyset&\text{otherwise}
							\end{cases}
\end{align*}
We claim that $Z = XY$. To prove the claim, consider for each set $X =
X_1 - X_2 + {} \dotsm {} \pm X_n$ associated with the decreasing sequence
$X_1$, \ldots, $X_n$ of subsets of $E$, the function $\mu_X$ defined
on $E$ by
$$
  \mu_X(x) = \max\,\{i\geq 1\mid x\in X_i\}
$$
with the convention that $\mu_X(x) = 0$ if $x \in E - X_1$. Then $x\in X$ if and only
if $\mu_X(x)$ is odd. We now evaluate $\mu_Z(x)$ as a function of $i = \mu_X(x)$ and
$j = \mu_Y(x)$. We first observe that if $k \geq i + j$, then $x \notin Z_k$. Next,
if $i$ and $j$ are not both even, then $x \in X_iY_j$ and $X_iY_j \subseteq
Z_{i+j-1}$, whence $\mu_Z(x) = i + j -1$. Finally, if $i$ and $j$ are both even, then
$x \notin Z_{i+j-1}$ and thus $\mu_Z(x)$ is either equal to $0$ or to $i+j-2$.
Summarizing the different cases, we observe that $\mu_X(x)$ and $\mu_Y(x)$ are both
odd if and only if $\mu_Z(x)$ is odd, which proves the claim. It follows that
$\cB(\cF)$ is closed under intersection.
\end{proof}

\noindent An equivalent definition of $\cB_n(\cF)$ was given by Hausdorff
\cite{Hausdorff57}. Let $X \dif Y$ denote the symmetric difference of two subsets $X$
and $Y$ of $E$.

\begin{prop}\label{prop:equivalent diff}
	For every $n \geq 0$, $\cB_n(\cF) = \{X_1 \dif X_2 \dif {} \dotsm {}
	\dif X_n \mid X_i\in \cF\}$.
\end{prop}

\begin{proof} Indeed, if $X = X_1 - X_2 + {} \dotsm {} \pm X_n$ with $X_1 \supseteq
X_2\supseteq X_3 \supseteq {}\dotsm{} \supseteq X_n$, then $X= X_1 \dif X_2 \dif {}
\dotsm {} \dif X_n$. In the opposite direction, if $X = X_1 \dif X_2 \dif {} \dotsm
{} \dif X_n$, then $X = Z_1 - Z_2 + {} \dotsm {} \pm Z_n$ where $Z_k =
\sum_{\text{$i_1, \ldots, i_k$ distincts}} X_{i_1} \dotsm X_{i_k}$.
\end{proof}


\section{Closure operators}\label{sec:Closure operators}

We review in this section the definition and the basic properties of
closure operators.

Let $E$ be a set. A map $X \to \clos{X}$ from $\cP(E)$ to itself is a \emph{closure
operator} if it is \emph{extensive, idempotent and isotone}, that is, if the
following properties hold for all $X, Y\subseteq E$:
\begin{conditions}
	\item $X\subseteq\clos{X}$ (extensive)

	\item $\clos{\clos{X}} = \clos{X}$ (idempotent)

	\item $X\subseteq Y$ implies $\clos{X}\subseteq\clos{Y}$ (isotone)
\end{conditions}
A set $F\subseteq E$ is \emph{closed} if $\clos{F} = F$. If $F$ is closed, and if
$X\subseteq F$, then $\clos{X}\subseteq \clos{F} = F$. It follows that $\clos{X}$ is
the least closed set containing $X$. This justifies the terminology ``closure''.
Actually, closure operators can be characterised by their closed sets.

\begin{prop}\label{prop:Closure operators}
	A set of closed subsets for some closure operator on $E$ is closed under
	(possibly infinite) intersection. Moreover, any set of subsets of $E$ closed under
	(possibly infinite) intersection is the set of closed sets for some closure
	operator.
\end{prop}

\begin{proof} Let $X\to \clos{X}$ be a closure operator and let $(F_i)_{i\in I}$ be a
family of closed subsets of $E$. Since a closure is isotone, $\clos{\bigcap_{i\in
I}F_i} \subseteq \clos{F_i} = F_i$. It follows that $\clos{\bigcap_{i\in I}F_i}
\subseteq \bigcap_{i\in I}F_i$ and thus $\bigcap_{i\in I}F_i$ is closed.

Given a set $\cF$ of subsets of $E$ closed under intersection, denote by
$\clos{X}$ the intersection of all elements of $\cF$ containing $X$. Then the map
$X\to \clos{X}$ is a closure operator for which $\cF$ is the set of closed sets.
\end{proof}

In particular, $\clos{X \cap Y}\subseteq\clos{X} \cap \clos{Y}$, but the
inclusion may be strict.

\begin{exa}
The trivial closure is the application defined by
$$
\clos{X} = \begin{cases}
	\emptyset &\text{if $X = \emptyset$}\\
         E  &\text{otherwise}
				 \end{cases}
$$
For this closure, the only closed sets are the empty set and $E$.
\end{exa}

\begin{exa}
	If $E$ is a topological space, the closure in the topological sense
	is a closure operator.
\end{exa}

\begin{exa}
	The convex hull is a closure operator. However, it is not induced by
	any topology, since the union of two convex sets is not necessarily
	convex.
\end{exa}

The \emph{intersection} of two closure operators $X \to
\closI{X}$ and $X \to \closII{X}$ is the function $X \to
\closIII{X}$ defined by $\closIII{X} = \closI{X} \cap \closII{X}$.

\begin{prop}\
   The intersection of two closure operators is a closure operator.
\end{prop}

\begin{proof} Let $\closIII{\phantom{X}}$ be the intersection of
$\closI{\phantom{X}}$ and $\closII{\phantom{X}}$. First, since
$X\subseteq\closI{X}$ and $X\subseteq\closII{X}$, one has $X\subseteq\closIII{X} =
\closI{X} \cap \closII{X}$. In particular, $\closIII{X} \subseteq
\closIII{\closIII{X}}$. Secondly, since $\closI{X} \cap \closII{X}
\subseteq \closI{X}$, $\closI{\closI{X} \cap \closII{X}} \subseteq
\closI{\closI{X}} = \closI{X}$. Similarly, $\closII{\closI{X} \cap
\closII{X}} \subseteq \closII{X}$. It follows that
\[
	\closIII{\closIII{X}} = \closI{\closI{X} \cap \closII{X}} \cap \closII{\closI{X}
	\cap \closII{X}} \subseteq \closI{X} \cap \closII{X} = \closIII{X}
\]
and hence $\closIII{X} = \closIII{\closIII{X}}$. Finally, if
$X \subseteq Y$, then $\closI{X} \subseteq \closI{Y}$ and $\closII{X}
\subseteq \closII{Y}$, and therefore $\closIII{X} \subseteq
\closIII{Y}$.\end{proof}

\noindent Let us conclude this section by giving a few examples of closure operators
occurring in the theory of formal languages.


\begin{exa}\label{ex:Iteration} 
\textbf{Iteration}. The map $L \to L^*$ is a closure operator. Similarly, the map $L
\to L^+$, where $L^+$ denotes the subsemigroup generated by $L$, is a closure
operator.
\end{exa}


\begin{exa}\label{ex:Shuffle} \textbf{Shuffle ideal}. 
The \emph{shuffle product} (or simply \emph{shuffle}) of two languages $L_1$ and
$L_2$ over $A$ is the language
\begin{multline*}
	L_1 \shuffle L_2 = \{ w \in A^* \mid w = u_1v_1 {} \dotsm {} u_nv_n \text{ for
	some words $u_1, \ldots, u_n, v_1, \ldots, v_n$ of $A^*$} \\
	\text{such that $u_1 {} \dotsm {} u_n \in L_1$ and $v_1 {} \dotsm {} v_n \in L_2$}
	\}\,.
\end{multline*}
The shuffle product defines a commutative and associative operation over the set of
languages over $A$. Given a language $L$, the language $L \shuffle A^*$ is called the
\emph{shuffle ideal} generated by $L$ and it is easy to see that the map $L \to L
\shuffle A^*$ is a closure operator.

This closure operator can be extended to infinite words in two ways: the \emph{finite
and infinite shuffle ideals} generated by an $\omega$-language $X$ are respectively:
\begin{align*}
	X \shuffle A^* &= \{ y_0x_1y_1 {} \dotsm {} x_ny_nx \mid y_0, \dots, y_n \in A^*
	\text{ and } x_1{} \dotsm {} x_n x \in X \} \\
	X \shuffle A^\omega &= \{ y_0x_1y_1x_2 {} \dotsm {} \mid y_0, \dots, y_n, \dots \in A^*
	\text{ and } x_1x_2{} \dotsm {} \in X \}
\end{align*}
The maps $X \to X \shuffle A^*$ and $X \to X \shuffle A^\omega$ are both closure
operators.
\end{exa}

\begin{exa}\label{ex:Ultimate closure} \textbf{Ultimate closure}. The
\emph{ultimate closure} of a language $X$ of infinite words is defined by:
\[
	\Ult(X) = \{ ux \mid u\in A^* \text{ and } vx \in X \text{ for some } v \in A^*\}
\]
The map $X \to \Ult(X)$ is a closure operator.
\end{exa}


\section{Approximation}\label{sec:Approximation}

In this section, we consider a set $\cF$ of closed sets of $E$ containing the empty
set. It follows that the corresponding closure operator satisfies the condition
$\clos{\emptyset} = \emptyset$. We first define the notion of an \emph{approximation}
of a set by a chain of closed sets. Then the existence of a best approximation will
be established.
In this section, $L$ is a subset of $E$.

\begin{defi}
A chain $F_1 \supseteq F_2 \supseteq {} \dotsm {} \supseteq F_n$ of closed sets is an
\emph{$n$-approximation of $L$} if the following inclusions hold for all $k$ such 
that $2k + 1 \leq n$:
\begin{multline*}
	F_1 - F_2 \subseteq F_1 - F_2 + F_3 - F_4 \subseteq {} \dotsm {} \subseteq F_1 -
	F_2 + {} \dotsm {} + F_{2k-1} - F_{2k}\subseteq {} \dotsm {} \\
	\subseteq L\subseteq {} \dotsm {} \subseteq F_1 - F_2 + F_3 - {} \dotsm {} + F_{2k+1}
	\subseteq {} \dotsm {} \subseteq F_1 - F_2 + F_3 \subseteq F_1
\end{multline*}
\end{defi}
\noindent There is a natural order among the $n$-approximations of a given set $L$.
An $n$-approximation $F_1 \supseteq F_2 \supseteq {} \dotsm {} \supseteq F_n$ of $L$ is
said to be \emph{better} than an $n$-approximation $F'_1 \supseteq F'_2 \supseteq
{} \dotsm {} \supseteq F'_{n}$ if, for all $k$ such that $2k+1 \leq n$,
\begin{align*}
	F_1 - F_2 + F_3 - {} \dotsm {} + F_{2k+1} &\subseteq F'_1 - F'_2 + F'_3 -
	{} \dotsm {} + F'_{2k+1}\\
\noalign{\noindent and}
	F'_1 - F'_2 + {} \dotsm {} + F'_{2k-1} - F'_{2k} &\subseteq F_1 - F_2 +
	{} \dotsm {} + F_{2k-1} - F_{2k}
\end{align*}
We will need the following elementary lemma:

\begin{lem}\label{lem:elementary}
	Let $X$, $Y$ and $Z$ be subsets of $E$.
	\begin{conditions}
		\item \label{item:Z-X} The conditions $X-Y \subseteq Z$ and $X-Z\subseteq Y$ are
		equivalent.
		
		\item \label{item:X+Y} The conditions $Z \subseteq X+Y$ and ${X^c \cap Z}\subseteq
		Y$ are equivalent.
		
		\item \label{item:XYZ} If $Y \subseteq X$ and $X-Y \subseteq Z$, then $X-Z = 
		Y-Z$ and $X + Z = Y + Z$.
	\end{conditions}
\end{lem}

\noindent The description of the best approximation of $L$ requires the introduction
of two auxiliary functions. For every subset $X$ of $E$, set
\begin{equation}\label{eq:f and g}
  f(X) = \clos{X-L} \quad\text{and}\quad g(X) = \clos{X \cap L}
\end{equation}
The key properties of these functions are formulated in the following lemma.

\begin{lem}\label{lem:funcfg}
	The following properties hold for all subsets $X$ and $Y$ of $E$:
\begin{conditions}
	\item \label{item:fg} $X - f(X) \subseteq L$ and $L \subseteq X + g(X^c)$,
	 
	\item \label{item:XYf} if $X \supseteq Y \supseteq L$, then $f(X) \supseteq f(Y)$
	and $X -f(X) \subseteq Y - f(Y) \subseteq L$,

	 \item \label{item:XYg} if $X \subseteq Y \subseteq L$, then $g(X) \subseteq g(Y)$
	 and $L \subseteq Y + g(Y^c) \subseteq X + g(X^c)$.
\end{conditions}
\end{lem}

\begin{proof} Let $X$ and $Y$ be subsets of $E$.

\noindent \eqref{item:fg} follows from a simple computation: 
\begin{align*}
 	X - f(X) &= X - \clos{X-L} \subseteq X - (X-L) = X\cap L \subseteq L \\
	X + g(X^c) &= X + \clos{X^c \cap L} \supseteq (X \cap L) + (X^c \cap L) = L.
\end{align*}
\eqref{item:XYf} Suppose that $X \supseteq Y \supseteq L$. Then $X - L \supseteq Y -
L$ and thus $\clos{X - L} \supseteq \clos{Y - L}$, that is, $f(X) \supseteq f(Y)$.
Furthermore, $X - Y \subseteq X - L \subseteq \clos{X - L} = f(X)$. Applying part
\eqref{item:XYZ} of Lemma \ref{lem:elementary} with $Z = f(X)$, one gets $X - f(X) =
Y - f(X)$, whence $X - f(X)\subseteq Y - f(Y)$ since $f(X) \supseteq f(Y)$ by the
first part of \eqref{item:XYf}.

\noindent \eqref{item:XYg} Suppose that $X \subseteq Y \subseteq L$. Then ${X \cap L}
\subseteq {Y \cap L}$ and thus $g(X) \subseteq g(Y)$. Furthermore, $Y - X = X^c \cap
Y \subseteq X^c \cap L \subseteq \clos{X^c \cap L} = g(X^c)$. Applying part
\eqref{item:XYZ} of Lemma \ref{lem:elementary} with $Z = g(X^c)$, one gets $Y +
g(X^c)= X + g(X^c)$, whence $Y + g(Y^c) \subseteq X + g(X^c)$ since $g(Y^c) \subseteq
g(X^c)$ by the first part of \eqref{item:XYg}.
\end{proof}

\begin{lem} \label{lem:funcfg2}
	Let $F_1 \supseteq F_2 \supseteq {} \dotsm {} \supseteq F_n$ be an
	$n$-approximation of $L$ and, for $1 \leq k \leq n$, let $S_k = F_1 - F_2 + {}
	\dotsm {} \pm F_k$. Then, for $1 \leq k \leq n$,
\begin{equation}
\left\{
\begin{aligned}\label{eq:fSk}
	f(S_k) &= f(F_k) \text{ if $k$ is odd} \\
	g(S_k^c) &= g(F_k) \text{ if $k$ is even}
\end{aligned}
\right.
\end{equation}
\end{lem}

\begin{proof} If $k=1$, then $S_1=F_1$ and the result is trivial. Suppose that $k>
1$. If $k$ is odd, $S_{k-1} \subseteq L$ and thus $S_k - L = (S_{k-1} + F_k) - L =
F_k - L$. It follows that $f(S_k) = f(F_k)$. If $k$ is even, $L \subseteq S_{k-1}$
and thus $S_k^c \cap L = (S_{k-1}^c + F_k) \cap L = F_k \cap L$. Therefore $g(S_k^c)
= g(F_k)$.\end{proof}

\noindent Define a sequence $(L_n)_{n\geq 0}$ of subsets of $E$ by
$L_0= E$ and, for all $n \geq 0$,
\begin{equation}\label{eq:fg}
  L_{n+1} =\begin{cases}
		f(L_n) &\text{if $n$ is odd}\\
    g(L_n) &\text{if $n$ is even}
		\end{cases}
\end{equation}
The next theorem expresses the fact that the sequence $(L_n)_{n\geq 0}$ is the best
approximation of $L$ as a Boolean combination of closed sets. In particular, if $L_n
= \emptyset$ for some $n > 0$, then $L \in \cB_{n-1}(\cF)$.

\begin{thm}\label{thm:best $n$-approximation}
	Let $L$ be a subset of $E$. For every $n > 0$, the sequence $(L_k)_{1 \leq k\leq
	n}$ is the best $n$-approximation of $L$.
\end{thm}

\begin{proof} We first show that the sequence $(L_k)_{1 \leq k\leq n}$ is an
$n$-approximation of $L$. First, every $L_k$ is closed by construction. We show that
$L_{k+1} \subseteq L_k$ by induction on $k$. This is true for $k=0$ since $L_0 = E$.
Now, if $k$ is even, $L_{k+1} = \clos{L_k \cap L} \subseteq \clos{L_k} = L_k$ and if
$k$ is odd, $L_{k+1} = \clos{L_k - L} \subseteq \clos{L_k} = L_k$.

Set, for $k > 0$, $S_k = L_1 - L_2 + {} \dotsm {} \pm L_{k}$. By part \eqref{item:fg}
of Lemma \ref{lem:funcfg}, the relations $L_{2k-1} - L_{2k} = L_{2k-1} - f(L_{2k-1})
\subseteq L$ hold for every $k> 0$, and similarly, $L_{2k} - L_{2k+1} = L_{2k} -
g(L_{2k}) \subseteq L^c$. It follows that $S_{2k} \subseteq L$. Furthermore
$S_{2k+1}^c = (L_0 -L_1) + (L_2 -L_3) + {} \dotsm {} + (L_{2k} - L_{2k+1}) \subseteq
L^c$ and thus $L\subseteq S_{2k+1}$.

We now show that the sequence $(L_k)_{1 \leq k\leq n}$ is the best approximation of
$L$. Let $(L'_k)_{1 \leq k \leq n}$ be another $n$-approximation of $L$. Set, for $k
> 0$, $S'_k = L'_1 - L'_2 + {} \dotsm {} \pm L'_{k}$. Then, by definition, $L
\subseteq L'_1$ and thus
$$
  S_1 = L_1 = \clos{L} \subseteq \clos{L'_1} = L'_1 = S'_1.
$$
Let $k > 0$. Suppose by induction that $S_{2k-1} \subseteq S'_{2k-1}$. We show
successively that $S_{2k} \subseteq S'_{2k}$ and $S_{2k+1} \subseteq S'_{2k+1}$.

By definition of an approximation, $S'_{2k} = {S'_{2k-1} - L'_{2k}} \subseteq L$, and
thus ${S'_{2k-1} - L} \subseteq L'_{2k}$ by part \eqref{item:Z-X} of Lemma
\ref{lem:elementary}. It follows that $f(S'_{2k-1}) = \clos{S'_{2k-1} - L} \subseteq
\clos{L'_{2k}} = L'_{2k}$. Now, since $S'_{2k-1} \supseteq S_{2k-1} \supseteq L$, one
can apply part \eqref{item:XYf} of Lemma \ref{lem:funcfg} to get
$$
	S'_{2k} = {S'_{2k-1} - L'_{2k}} \subseteq {S'_{2k-1} - f(S'_{2k-1})} \subseteq
	{S_{2k-1} - f(S_{2k-1})}.
$$
Moreover since $f(S_{2k-1}) = f(L_{2k-1}) = L_{2k}$ by Lemma \ref{lem:funcfg2}, one
gets 
$$
	S'_{2k} \subseteq S_{2k-1} - f(S_{2k-1}) = S_{2k-1} - L_{2k} = S_{2k}.
$$
Similarly, $L \subseteq S'_{2k+1} = S'_{2k} + L'_{2k+1}$ and hence ${(S'_{2k})^c \cap
L} \subseteq L'_{2k+1}$ by part \eqref{item:X+Y} of Lemma \ref{lem:elementary}. It
follows that $g((S'_{2k})^c) = \clos{(S'_{2k})^c \cap L} \subseteq \clos{L'_{2k+1}} =
L'_{2k+1}$. Now, since $S'_{2k} \subseteq S_{2k} \subseteq L$, one can apply part
\eqref{item:XYg} of Lemma \ref{lem:funcfg} to get
$$
	S_{2k} + g(S_{2k}^c) \subseteq S'_{2k} + g((S'_{2k})^c) \subseteq S'_{2k}
	+ L'_{2k+1} = S'_{2k+1}.
$$
Moreover since the equalities $g(S_{2k}^c) = g(L_{2k}) = L_{2k+1}$ hold by Lemma
\ref{lem:funcfg2}, one gets
$$
  S_{2k+1} = S_{2k} + L_{2k+1} = S_{2k} + g(S_{2k}^c) \subseteq S'_{2k+1}. 
$$
which concludes the proof.
\end{proof}  

When $\cF$ is a set of subsets of $E$ closed under arbitrary intersection, Theorem
\ref{thm:best $n$-approximation} provides a characterization of the classes
$\cB_n(\cF)$.

\begin{cor}\label{cor:Bn(F)}
	Let $L$ be a subset of $E$ and let $\cF$ be a set of subsets of $E$ closed under
	(possibly infinite) intersection and containing the empty set. Let $(L_k)_{1 \leq k
	\leq n}$ be the best $n$-approximation of $L$ with respect to $\cF$. Then $L \in
	\cB_{n-1}(\cF)$ if and only if $L_n = \emptyset$ and in this case
\begin{equation}\label{eq:L diff F}
	L = L_1 - L_2 + {} \dotsm {} \pm L_{n-1}
\end{equation}
\end{cor}

\begin{proof} If $L \in \cB_{n-1}(\cF)$, then $L = F_1 - F_2 + {} \dotsm {} \pm
F_{n-1}$ with $F_1, \ldots, F_{n-1} \in \cF$. Let $F_n = \emptyset$. Then the
sequence $(F_k)_{1\leq k \leq n}$ is an $n$-approximation of $L$. Since $(L_k)_{1
\leq k \leq n}$ is the best $n$-approximation of $L$, one has $L =L_1 - L_2 + {}
\dotsm {} \pm L_{n-1}$. Thus, with the notation of Lemma \ref{lem:funcfg2},
\begin{equation}\label{eq:Ln}
\left\{
\begin{aligned}
	f(L_{n-1}) &= f(L) = \emptyset \text{ if $n-1$ is odd} \\
	g(L_{n-1}) &= g(L^c) = \emptyset \text{ if $n-1$ is even}
\end{aligned}
\right.
\end{equation}
Therefore, $L_n = \emptyset$ by (\ref{eq:fg}).

\noindent Conversely, suppose that $L_n = \emptyset$. If $n = 2k$, then
\begin{align*}
	&(L_1 - L_2) + {} \dotsm {} + (L_{2k-1} - L_{2k}) \subseteq L   
	\subseteq (L_1 - L_2) + {} \dotsm {} + (L_{2k-3} - L_{2k-2}) + L_{2k-1} \\
\noalign{\noindent If $n = 2k + 1$, then}
	&(L_1 - L_2) + {} \dotsm {} + (L_{2k-1} - L_{2k}) \subseteq L  
	\subseteq (L_1 - L_2) + {} \dotsm {} + (L_{2k-1} - L_{2k}) + L_{2k+1} 
\end{align*}
In both cases, one gets $L = L_1 - L_2 + {} \dotsm {} \pm L_{n-1}$ and thus
$L \in \cB_{n-1}(\cF)$.
\end{proof}

\noindent Let us illustrate this corollary by a concrete example.

\begin{exa}\label{ex:best approximation}
Let $A = \{a, b, c\}$ and let $\cL$ be the lattice of shuffle ideals. If $L$ is the
language $\{1, a, b, c, ab, bc, abc\}$, a straightforward computation gives 
\begin{align*}
	L_0 &= A^*  \\
	L_1 &= g(L_0) = A^* \shuffle (L_0 \cap L) = A^* \shuffle L = A^* \\
	L_2 &= f(L_1) = A^* \shuffle (L_1 - L) = A^* \shuffle \{aa, ac, ba, bb, ca, cb,
	cc\} = A^* - \{1, a, b, c, ab, bc\} \\
	L_3 &= g(L_2) = A^* \shuffle (L_2 \cap L) = A^* \shuffle abc \\
	L_4 &= f(L_3) = A^* \shuffle (L_3 - L) = A^* \shuffle ((A^* \shuffle abc) - abc)\\
	&= A^* \shuffle \{aabc, abac, abca, babc,
	abbc, abcb, cabc, acbc, abcc \} \\
	L_5 &= g(L_4) = A^* \shuffle (L_4 \cap L) = \emptyset
\end{align*}
It follows that $L = L_1 - L_2 + L_3 - L_4$ and $L \in \cB_4(\cL)$, but $L \notin 
\cB_3(\cL)$.
\end{exa}

It is also possible to use the approximation algorithm for a set $\cL$ of subsets of 
$E$ closed under (possibly infinite) union and containing the set $E$. In this case, 
the set 
\[
  \cL^c = \{L^c \mid L \in \cL\}
\]
is closed under (possibly infinite) intersection and contains the empty set.
Consequently, the approximation algorithm can be applied to $\cL^c$ but it describes
the difference hierarchy $\cB_{n}(\cL^c)$. To recover the difference hierarchy
$\cB_{n}(\cL)$, the following algorithm can be used. First compute the best
$\cL^c$-approximation of even length of $L$ and the best $\cL^c$-approximation of odd
length of $L^c$, say
\begin{align}
	L   &= L_1^c - L_2^c + {} \dotsm {} \pm L_n^c    \label{eq:Lc decomposition of L} \\
	L^c &= F_1^c - F_2^c + {} \dotsm {} \pm F_m^c  \label{eq:Lc decomposition of LC} 
\end{align}
with $n$ even, $m$ odd, $L_i, F_i \in \cL$ and $L_n$ and $F_m$ possibly empty to fill
the parity requirements. Now $L$ admits the following $\cL$-decompositions, where 
$L_1$ and $F_1$ are possibly empty (and consequently deleted):
\begin{align}
	L &= L_n - L_{n-1} + {} \dotsm {} \pm L_1  \label{eq:L decomposition of L} \\
	  &= F_m - F_{m-1} + {} \dotsm {} \pm F_1  \label{eq:L decomposition of Lc}	
\end{align}
It remains to take the shortest of the two expressions to get the best
$\cL$-approximation of $L$.


\section{Decidability questions on regular languages}\label{sec:Decidabilty questions}

Given a lattice of regular languages $\cL$, four decidability questions arise:

 
\begin{qu}\label{Q:L decidable}
	Is the membership problem for $\cL$ decidable? 
\end{qu}

\begin{qu}\label{Q:BL decidable}
	Is the membership problem for $\cB(\cL)$ decidable? 
\end{qu}

\begin{qu}\label{Q:BnL decidable}
	For a given positive integer $n$, is the membership problem for $\cB_n(\cL)$
	decidable?
\end{qu}

\begin{qu}\label{Q:hierarchy decidable}
	Is the hierarchy $\cB_n(\cL)$ decidable?
\end{qu}

\noindent In other words, given a regular language $L$, Question \ref{Q:L decidable}
asks to decide whether $L \in \cL$, Question \ref{Q:BL decidable} whether $L \in
\cB(\cL)$ and Question \ref{Q:BnL decidable} whether $L \in \cB_n(\cL)$. Question
\ref{Q:hierarchy decidable} asks whether one can one effectively compute the smallest
$n$ such that $L \in \cB_n(\cL)$, if it exists. Note that if Questions \ref{Q:BL
decidable} and \ref{Q:BnL decidable} are decidable, then so is Question
\ref{Q:hierarchy decidable}. Indeed, given a language $L$, one first decides whether
$L$ belongs to $\cB(\cL)$ by Question \ref{Q:BL decidable}. If the answer is
positive, this ensures that $L$ belongs to $\cB_n(\cL)$ for some $n$ and Question
\ref{Q:BnL decidable} allows one to find the smallest such $n$.

If the lattice $\cL$ is finite, it is easy to solve the four questions in a positive
way. In some cases, a simple application of Corollary \ref{cor:Bn(F)} suffices to solve  
Question \ref{Q:BnL decidable} immediately. One just needs to find the appropriate 
closure operator and to provide algorithms to compute the functions $f(X)$ and $g(X)$ 
defined by \eqref{eq:f and g}.

\begin{exa}\label{ex:B*}
	Let $\cL$ be the lattice generated by the languages of the form $B^*$, where $B
	\subseteq A$. Then both $\cL$ and $\cB(\cL)$ are finite. It is known that a regular
	language belongs to $\cL$ if and only if its syntactic ordered monoid is idempotent
	and commutative and satisfies the inequation $1 \leq x$ for all $x$ \cite{Pin95b}.
	It belongs to $\cB(\cL)$ if and only if its syntactic monoid is idempotent and
	commutative.
	
	Finally, one can define a closure operator by setting $\overline{L} = B^*$, where
	$B$ is the set of letters occurring in some word of $L$. For instance, let $L =
	(\{a,b,c\}^* - \{b,c\}^*) + (\{a,b\}^* - a^*) + 1$. This language belongs to
	$\cB(\cL)$ and its minimal automaton is represented below:
\begin{center}
  \includegraphics{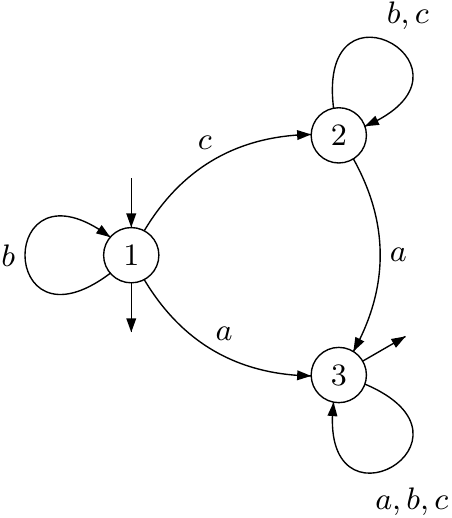}
\end{center}
Applying the approximation algorithm of Section \ref{sec:Approximation}, one gets
$L_0 = \{a,b,c\}^*$, $L_1 = \{b,c\}^*$, $L_2 = b^*$ and $L_3 = \emptyset$ and thus $L
= \{a,b,c\}^* - \{b,c\}^* + b^*$ is the best $3$-approximation of $L$.
\end{exa}
If the lattice is infinite, our four questions become usually much harder, but can
still be solved in some particular cases. We will discuss this in Sections
\ref{sec:polynomial closure} and \ref{sec:Cyclic}, but first present a powerful tool
introduced in \cite{Carton97}, chains in ordered monoids.


\section{Chains and difference hierarchies}\label{sec:Chains and difference
hierarchies}

Chains can be defined on any ordered set. We first give their definition, then
establish a connection with difference hierarchies.


\begin{defi}\label{def:chain}
Let $(E, \leq)$ be a partially ordered set and let $X$ be a subset of $E$. A
\emph{chain} of $E$ is a strictly increasing sequence
\[
	x_0 < x_1 < \ldots < x_{m-1} 
\] 
of elements of $E$. It is called an \emph{$X$-chain} if $x_0$ is in $X$ and the
$x_i$'s are alternatively elements of $X$ and of its complement $X^c$. The integer
$m$ is called the \emph{length} of the chain. We let $m(X)$ denote the maximal length
of an $X$-chain.
\end{defi}

\noindent There is a subtle connection between chains and difference hierarchies of
regular languages. Let $M$ be a finite ordered monoid and let $\varphi : A^*
\rightarrow M$ be a surjective monoid morphism. Let
\[
  \cL = \{ \varphi^{-1}(U) \mid \text{$U$ is an upper set of $M$} \}
\]
By definition, every language of $\cL$ is recognised by the ordered monoid $M$.

\begin{prop}\label{prop:chains to difference}
	If there exists a subset $P$ of $M$ such that $L= \varphi^{-1}(P)$ and $m(P) \leq
	n$, then $L$ belongs to $\cB_n(\cL)$.
\end{prop}

\noindent Before starting the proof, let us clarify a delicate point. The condition
$L= \varphi^{-1}(P)$ means that $L$ is recognised by the \emph{monoid $M$}. It does
not mean that $L$ is recognised by the \emph{ordered monoid $M$}, a property which
would require $P$ to be an upper set.

\begin{proof}
For each $s \in M$, let $m(P, s)$ be the maximal length of a $P$-chain ending with
$s$. Finally, let, for each $k > 0$,
$$
	U_k = \{s \in M \mid m(P, s) \geq k \}
$$
We claim that $U_k$ is an upper set of $M$. Indeed, if $s \in U_k$, there exists a
$P$-chain $x_0 < x_1 < {} \dotsm {} < x_{r-1} = s$ of length $r \geq k$. Let $t$ be an
element of $M$ such that $s \leq t$. If $s$ and $t$ are not simultaneously in $P$,
then $x_0 < x_1 < {} \dotsm {} < x_{r-1} < t$ is a $P$-chain of length $r+1 \geq k$.
Otherwise, $x_0 < x_1 < {} \dotsm {} < x_{r-2} < t$ is a $P$-chain of length $r \geq k$.
Thus $m(P, t) \geq k$, and $t \in U_k$, proving the claim.

We now show that
\begin{equation}\label{eq:P =}
	P = U_1 - U_2 + U_3 - U_4 {} \dotsm {} \pm U_n  
\end{equation}
First observe that $s \in P$ if and only if $m(P, s)$ is odd. Since $m(P) \leq
n$, one has $m(P, s) \leq n$ for every $s \in M$ and thus $U_{n+1} =
\emptyset$. Formula \eqref{eq:P =} follows, since for each $r \geq 0$, 
\[
	\{s \in M \mid m(P, s) = r \} = U_r - U_{r+1}.
\]
Let, for $1 \leq i\leq n$, $L_i = \varphi^{-1}(U_i)$. Since $U_i$ is
an upper set, each $L_i$ belongs to $\cL$. Moreover, one gets from \eqref{eq:P =}
the formula
\begin{equation}\label{eq:L =}
	L = L_1 - L_2 + L_3 {} \dotsm {} \pm L_n  
\end{equation}
which shows that $L \in \cB_n(\cL)$.
\end{proof}

\noindent We now establish a partial converse to Proposition \ref{prop:chains to
difference}. A \emph{lattice of regular languages} is a set $\cL$ of regular
languages of $A^*$ containing $\emptyset$ and $A^*$ and closed under finite union and
finite intersection.

\begin{prop}\label{prop:BnV}
	Let $\cL$ be a lattice of regular languages. If a language $L$ belongs to
	$\cB_n(\cL)$, then there exist an ordered stamp $\eta : A^* \rightarrow M$ and a
	subset $P$ of $M$ satisfying the following conditions:
	\begin{conditions} 
	\item \label{item:restricted product} $\eta$ is a restricted product of 
	syntactic ordered stamps of members of $\cL$,
	
	\item \label{item:phi-1 P} $L= \eta^{-1}(P)$,
	
	\item \label{item:m(P)} $m(P) \leq n$.
\end{conditions}
\end{prop}

\begin{proof} If $L \in \cB_n(\cL)$, then
\[
	L = L_1 - L_2 + L_3 \ \dotsm\ \pm L_n
\]
with $L_1 \supseteq L_2 \supseteq {} \dotsm {} \supseteq L_n$ and $L_i \in \cL$. Let
$\eta_i: A^* \rightarrow (M_i, \leq_i)$ be the syntactic morphism of $L_i$ and let
$P_i = \eta_i(L_i)$. Then each $P_i$ is an upper set of $M_{i}$ and $L_i =
\eta_i^{-1}(P_i)$. Let $\eta : A^* \rightarrow M$ be the restricted product of the
stamps $\eta_i$. Condition \eqref{item:restricted product} is satisfied by
construction.

Observe that if $\eta(u) = (s_1, \ldots, s_n)$ is an element of $M$, the condition
$s_{i+1} \in P_{i+1}$ is equivalent to $u \in L_{i+1}$, and since $L_{i+1}$ is a
subset of $L_i$, this condition also implies $u\in L_i$ and $s_i \in P_i$.
Consequently, for each element $s = (s_1, \ldots, s_n)$ of $M$, there exists a unique
$k \in \{0, \ldots, n\}$ such that
\[
	s_1 \in P_1, \ldots, s_k \in P_k, s_{k+1}\notin P_{k+1}, \ldots, s_n \notin P_n
\]
This unique $k$ is called the \emph{cut} of $s$. Setting
\[
	P = \{s \in M \mid \text{the cut of $s$ is odd}\}
\]
one gets, with the convention $L_{n+1} = \emptyset$ for $n$ odd, 
\begin{equation}\label{eq:L}
	\eta^{-1}(P) = \bigcup_{\text{$k$ odd}} \Bigl((L_1 \cap {} \dotsm {} \cap L_k)
	- L_{k+1}\Bigr) = \bigcup_{\text{$k$ odd}} (L_k - L_{k+1})
	= L
\end{equation}
which proves \eqref{item:phi-1 P}.

Let now $x_0 < x_1 < {} \dotsm {} < x_{m-1}$ be a $P$-chain. Let, for $0 \leq i \leq
m-1$, $x_i = (s_{i,1}, \ldots, s_{i,n})$ and let $k_i$ be the cut of $x_i$. We claim
that $k_{i+1} > k_i$. Indeed, since $x_i < x_{i+1}$, $s_{i, k_i} \leq_i s_{i+1, k_i}$
and since $P_i$ is an upper set, $s_{i, k_i} \in P_i$ implies $s_{i+1, k_i}\in
P_{i+1}$, which proves that $k_{i+1} \geq k_i$. But since $x_i$ and $x_{i+1}$ are not
simultaneously in $P$, their cuts must be different, which proves the claim. Since
$x_0 \in P$, the cut of $x_0$ is odd, and in particular, non-zero. It follows that $0
< k_0 < k_1 < {} \dotsm {} < k_{m-1}$ and since the cuts are numbers between $0$ and
$n$, $m \leq n$, which proves \eqref{item:m(P)}.
\end{proof}

It is tempting to try to improve Proposition \ref{prop:BnV} by taking for $M$ the
syntactic morphism of $L$ and for $\varphi$ the syntactic morphism of $L$. However,
Example \ref{ex:best approximation} ruins this hope. Indeed, let $F = \{1, a, b, c,
ab, bc, abc\}$ be the set of factors of the word $abc$. Then the syntactic monoid of
$L$ can be defined as the set $F \cup \{0\}$ equipped with the product defined by
\[
  xy = 
  \begin{cases}
  	xy &\text{if $x$, $y$ and $xy$ are all in $F$} \\
  	0  &\text{otherwise}
  \end{cases}
\]
Now the syntactic image of $L$ is equal to $F$. It follows that $M - F = \{0\}$ and
thus, whatever order is taken on $M$, the length of a chain is bounded by $3$.
Nevertheless, if $\cL$ is the lattice of shuffle ideals, then $L$ does not belong to
$\cB_3(\cL)$.

Therefore, if $L$ is a regular language, the maximal length of an $L$-chain cannot be
in general computed in the syntactic monoid of $L$. It follows that decidability
questions on $\cB_n(\cL)$, as presented in Section \ref{sec:Decidabilty questions}
below, cannot in general be solved just by inspecting the syntactic monoid. An
exceptional case where the syntactic monoid suffices is presented in the next
section.


\section{The difference hierarchy of the polynomial closure of a
lattice}\label{sec:polynomial closure}

A language $L$ of $A^*$ is a \emph{marked product} of the languages $L_0, L_1,
\ldots, L_n$ if
\[
	L = L_0a_1L_1 \dotsm a_nL_n
\]
for some letters $a_1, \ldots, a_n$ of $A$. Given a set $\cL$ of languages, the
\emph{polynomial closure} of $\cL$ is the set of languages that are finite unions of
marked products of languages of $\cL$. The \emph{polynomial closure} of $\cL$ is
denoted $\PolL$ and the Boolean closure of $\PolL$ is denoted $\BPolL$. Finally, let
$\coPolL$ denote the set of complements of languages in $\PolL$. In this section, we 
are interested in the difference hierarchy induced by $\PolL$.
We consider several examples.


\subsection{Shuffle ideals}\label{subsec:Shuffle ideals}

If $\cL = \{\emptyset, A^*\}$, then $\PolL$ is exactly the set of shuffle ideals
considered in Examples \ref{ex:Shuffle} and \ref{ex:B*} and $\BPolL$ is the class of
\emph{piecewise testable languages}. The following easy result was mentioned in
\cite{Pin95b}.

\begin{prop}\label{prop:shuffle ideals}
	A language is a shuffle ideal if and only if its syntactic ordered monoid $M$
	satisfies the inequation $1 \leq x$ for all $x \in M$.
\end{prop}

The syntactic characterization of piecewise testable languages follows from a much
deeper result of Simon \cite{Simon75}.

\begin{thm}\label{thm:piecewise testable languages}
	A language is piecewise testable if and only if its syntactic monoid is
	$\cJ$-trivial.
\end{thm}

Note that the closed sets of the closure operator $X \to X \shuffle A^*$ of Example
\ref{ex:Shuffle} are exactly the shuffle ideals. It follows that for the lattice
$\cL$ of shuffle ideals, the four questions mentioned earlier have a positive answer.
More precisely, the decidability of the membership problem for $\cL$ and for
$\cB(\cL)$ follows from Proposition \ref{prop:shuffle ideals} and Theorem
\ref{thm:piecewise testable languages}, respectively. The decidability of Question
\ref{Q:BnL decidable} (and hence of Question \ref{Q:hierarchy decidable}) follows
from the approximation algorithm. See Example \ref{ex:best approximation}.


\subsection{Group languages}\label{subsec:Group languages}

Recall that a \textit{group language} is a language whose syntactic monoid is a
group, or, equivalently, is recognized by a finite deterministic automaton in which
each letter defines a permutation of the set of states. According to the definition
of a polynomial closure, a \textit{polynomial of group languages} is a finite union
of languages of the form $L_0a_1L_1 {} \dotsm {} a_kL_k$ where $a_1, \ldots, a_k$ are
letters and $L_0, \ldots, L_k$ are group languages.

Let $d_{\bG}$ be the metric on $A^*$ defined as follows:
\begin{align*}
	r_{\bG}(u,v) &= \min \left\{|M| \mid \text{$M$ is a finite group that separates $u$
	and $v$} \right\}\\
	d_{\bG}(u,v) &= 2^{-r_{\bG}(u,v)}
\end{align*}
\noindent It is known that $d_{\bG}$ defines the so-called \emph{pro-group topology}
on $A^*$. It is also known that the closure of a regular language for $d_{\bG}$ is
again regular and can be effectively computed. This result was actually proved in two
steps: it was first reduced to a group-theoretic conjecture in \cite{PinReutenauer91}
and this conjecture became a theorem in \cite{RibesZalesskii93}.

Let $\cG$ be the set of group languages on $A^*$ and let $\PolG$ be the polynomial
closure of $\cG$. We also let $\coPolG$ denote the set of complements of languages of
$\PolG$. The following characterization of $\coPolG$ was given in \cite{Pin94}.

\begin{thm}\label{thm:coPolG}
	Let $L$ be a regular language and let $M$ be its syntactic ordered monoid. The
	following conditions are equivalent:
\begin{samepage}
	\begin{conditions}	
	\item $L \in \coPolG$,
	
	\item $L$ is closed in the pro-group topology on $A^*$,
	
	\item for all $x \in M$, $x^\omega \leq 1$.
\end{conditions}	
\end{samepage}
\end{thm}

\noindent Theorem \ref{thm:coPolG} shows that $\coPolG$, and hence $\PolG$, is
decidable. The corresponding result for $\BPolG$ has a long story, related in detail
in \cite{Pin95c}, where several other characterizations can be found.

\begin{thm}\label{thm:BPolG}
	Let $L$ be a regular language and let $M$ be its syntactic monoid. The following
	conditions are equivalent:
\begin{conditions}	
	\item $L \in \BPolG$,
	
	\item the submonoid generated by the idempotents of $M$ is $\cJ$-trivial,
	
	\item for all idempotents $e$, $f$ of $M$, the condition $efe = e$ implies $ef = e
	= fe$.
\end{conditions}	
\end{thm}

\noindent We  now study the difference hierarchy based on $\coPolG$. Let
$\cF$ be the set of closed subsets for the pro-group topology.

\begin{prop}\label{prop:Bn(PolG)}
	For each $n \geq 0$, a regular language belongs to $\cB_n(\coPolG)$ if
	and only if it belongs to $\cB_n(\cF)$.
\end{prop}

\begin{proof} Theorem \ref{thm:coPolG} shows that $\coPolG$ is a subset of $\cF$. It 
follows that any language of $\cB_n(\coPolG)$ belongs to $\cB_n(\cF)$.

Let now $L$ be a regular language of $\cB_n(\cF)$ and let $(L_k)_{1 \leq k \leq n}$
be the best $n$-approximation of $L$ with respect to $\cF$. Corollary \ref{cor:Bn(F)}
shows that $L \in \cB_{n}(\cF)$ if and only if $L_{n+1} = \emptyset$. Moreover, in
this case $L = L_1 - L_2 + {} \dotsm {} \pm L_n$. According to the algorithm
described at the end of Section \ref{sec:Approximation}, the best $n$-approximation
of $L$ is obtained by alternating the two operations
\[
  f(X) = \clos{X-L} \quad\text{and}\quad g(X) = \clos{X \cap L}
\]
Now, as we have seen, the closure of a regular language for $d_{\bG}$ is regular. It
follows that if $X$ is regular, then both $f(X)$ and $g(X)$ are regular and closed.
By Theorem \ref{thm:coPolG}, they both belong to $\coPolG$. It follows that each
$L_k$ belongs to $\coPolG$ and thus $L \in \cB_n(\coPolG)$.
\end{proof}

This leads to the following corollary:

\begin{cor}\label{cor:Bn(PolG)}
	The difference hierarchy $\cB_n(\coPolG)$ is decidable.
\end{cor}

\begin{proof} Let $L$ be a regular language. Theorem \ref{thm:BPolG} shows that one
can effectively decide whether $L \in \BPolG$. If this is the case, it remains to
find the minimal $n$ such that $L \in \cB_n(\cF)$. But Proposition
\ref{prop:Bn(PolG)} shows that $L$ belongs to $\cB_n(\coPolG)$ if and only if it
belongs to $\cB_n(\cF)$. Moreover, since the closure of a regular language can be
effectively computed, the best $n$-approximation of $L$ with respect to $\cF$ can be
effectively computed. Now, Corollary \ref{cor:Bn(F)} gives an algorithm to decide
whether $L \in \cB_n(\cF)$.
\end{proof}


\section{Cyclic and strongly cyclic regular languages}\label{sec:Cyclic}

Cyclic and strongly cyclic regular languages are two classes of regular languages
related to symbolic dynamic and first studied in \cite{BealCartonReutenauer96}. It
was shown in \cite{Carton97} that an appropriate notion of chains suffices to
characterise the difference hierarchy based on the class of strongly cyclic regular
languages. This contrasts with Section \ref{sec:Chains and difference hierarchies},
in which the general results on chain did not lead to a full characterization of
difference hierarchies.

Let $\cA = (Q, A, \cdp)$ be a finite (possibly incomplete) deterministic automaton. A
word $u$ \emph{stabilises} a subset $P$ of $Q$ if $P \cdp u = P$. Given a subset $P$
of $Q$, let $\Stab(P)$ be the set of all words that stabilise $P$. The language
$\Stab(\cA)$ that stabilises $\cA$ is by definition the set of all words which
stabilise at least one nonempty subset of $Q$.

\begin{defi}\label{def:strongly cyclic}
	A language is \emph{strongly cyclic} if it stabilises some finite deterministic
	automaton.
\end{defi}

\begin{exa}\label{ex:strongly cyclic}
If $\cA$ is the automaton represented in Figure \ref{fig:strongly
cyclic}, then 
$$
	\Stab(\{1\}) = (b + aa)^*,\ \Stab(\{2\}) = (ab^*a)^*,\ \Stab(\{1,
	2\}) = a^*
$$ 
and $\Stab(\cA) = (b + aa)^* + (ab^*a)^* + a^*$.
\begin{figure}[H]
  \begin{center}
    \includegraphics{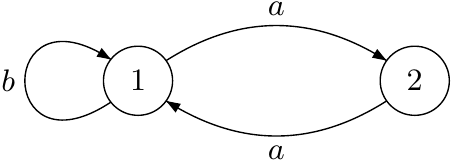}
  \end{center}
	\caption{The automaton $\cA$.}\label{fig:strongly cyclic}
\end{figure}
\end{exa}
\noindent One can show that the set of strongly cyclic languages of $A^*$ forms a
lattice of languages but is not closed under quotients. For instance, as shown in
Example \ref{ex:strongly cyclic}, the language $L = (b + aa)^* + (ab^*a)^* + a^*$ is
strongly cyclic, but Corollary \ref{cor:strongly cyclic are cyclic} will show that
its quotient $b^{-1}L = (b + aa)^*$ is not strongly cyclic, since $aa \in (b + aa)^*$
but $a \notin (b + aa)^*$.

We will also need the following characterization \cite[Proposition
7]{BealCartonReutenauer96}:

\begin{prop}\label{prop:charstabword}
	Let $\cA = (Q,A,E)$ be a deterministic automaton. A word~$u$ belongs to
	$\Stab(\cA)$ if and only if there is some state $q$ of~$\cA$ such that for every
	integer~$n$, the transition $q\cdot u^n$ exists.
\end{prop}

\noindent Strongly cyclic languages admit the following syntactic characterization
\cite[Theorem 8]{BealCartonReutenauer96}. As usual, $s^\omega$ denotes the 
idempotent power of $s$, which exists and is unique in any finite monoid.

\begin{prop}\label{prop:strongly cyclic languages}
	Let $L$ be a non-full regular language. The following conditions are equivalent:
	\begin{conditions}
		\item $L$ is strongly cyclic,
	
		\item there is a morphism $\varphi$ from $A^*$ onto a finite monoid $M$ with zero
		such that
\[
	L = \varphi^{-1}(\{s \in M \mid s^\omega \neq 0 \}),
\]

		\item the syntactic monoid $M$ of $L$ has a zero and the syntactic image of $L$
		is the set of all elements $s \in M$ such that $s^\omega \neq 0$.
	\end{conditions}
\end{prop}

\noindent Proposition \ref{prop:strongly cyclic languages} leads to a simple
syntactic characterization of strongly cyclic languages. Recall that a language of
$A^*$ is \emph{nondense} if there exists a word $u \in A^*$ such that $L \cap A^*uA^*
= \emptyset$.

\begin{prop}\label{prop:identities of strongly cyclic}
	Let $L$ be a regular language, let $M$ be its syntactic monoid and let $P$ be its
	syntactic image. Then $L$ is strongly cyclic if and only if it satisfies the
	following conditions, for all $u, x, v \in M$:
\begin{conditionsJE}{S}
	\item \label{eq:S1} $ux^\omega v \in P$ implies $x^\omega \in P$,
	
	\item \label{eq:S2} $x^\omega \in P$ if and only if $x \in P$.
\end{conditionsJE}
Furthermore, if these conditions are satisfied and if $L$ is not the full language,
then $L$ is nondense.
\end{prop}

\begin{proof}
Let $L$ be a strongly cyclic language, let $M$ be its syntactic monoid and let $P$ be
its syntactic image. If $L$ is the full language, then the conditions \ref{eq:S1} and
\ref{eq:S2} are trivially satisfied. If $L$ is not the full language, then
Proposition \ref{prop:strongly cyclic languages} shows that $M$ has a zero and that
$P = \{s \in M \mid s^\omega \neq 0 \}$. Observing that $x^\omega =
(x^\omega)^\omega$, one gets
\[
  x \in P \Longleftrightarrow x^\omega \neq 0 \Longleftrightarrow 
	(x^\omega)^\omega \neq 0 \Longleftrightarrow x^\omega \in P
\]
which proves \ref{eq:S2}. Similarly, one gets
\[
	ux^\omega v \in P \iff (ux^\omega v)^\omega \neq 0 \implies x^\omega \neq 0 \iff x 
	\in P
\]
which proves \ref{eq:S1}.

Conversely, suppose that $L$ satisfies \ref{eq:S1} and \ref{eq:S2}. If $L$ is full,
then $L$ is strongly cyclic. Otherwise, let $z \notin P$. Then $z^\omega \notin P$ by
\ref{eq:S1} and $uz^\omega v \notin P$ for all $u, v \in M$ by \ref{eq:S2}. This
means that $z$ is a zero of $M$ and that $0 \notin P$. By Proposition
\ref{prop:strongly cyclic languages}, it remains to prove that $x \in P$ if and only
if $x^\omega \neq 0$. First, if $x \in P$, then $x^\omega \in P$ by \ref{eq:S2} and
since $0 \notin P$, one has $x^\omega \neq 0$. Conversely, if $x^\omega \neq 0$, then
$ux^\omega v \in P$ for some $u, v \in M$, since $x^\omega$ is not equivalent to $0$
in the syntactic congruence of $P$. It follows that $x^\omega \in P$ by \ref{eq:S1}
and $x \in P$ by \ref{eq:S2}.
\end{proof}

\noindent We turn now to cyclic languages.

\begin{defi}\label{def:cyclic}
A subset of a monoid is said to be \emph{cyclic} if it is closed under conjugation,
power and root. That is, a subset $P$ of a monoid $M$ is cyclic if it satisfies the
following conditions, for all $u, v \in M$ and $n > 0$:
\begin{conditionsJE}{C}
	\item \label{item:power} $u^n \in P$ if and only if $u \in P$,

	\item \label{item:conjugation} $uv \in P$ if and only if $vu \in P$.
\end{conditionsJE}
\end{defi}
\noindent This definition applies in particular to the case of a language of $A^*$.

\begin{exa}\label{ex:cyclic}
  If $A = \{a, b\}$, the language $b^*$ and its complement $A^*aA^*$ are cyclic.
\end{exa}
\noindent One can show that regular cyclic languages are closed under inverses of
morphisms and under Boolean operations but not under quotients. For instance, the
language $L = \{abc, bca, cab\}$ is cyclic, but its quotient $a^{-1}L = \{bc\}$ is
not cyclic. Thus regular cyclic languages do not form a variety of languages.
However, they admit the following straightforward characterization in terms of
monoids.

\begin{prop} \label{prop:charcyclic}
	Let $L$ be a regular language of $A^*$, let $\varphi$ be a surjective morphism from
	$A^*$ to a finite monoid $M$ recognising $L$ and let $P = \varphi(L)$. Then $L$ is
	cyclic if and only if $P$ is cyclic.
\end{prop}

\begin{cor}\label{cor:strongly cyclic are cyclic}
	Every strongly cyclic language is cyclic. 
\end{cor}

\begin{proof}
	Let $L$ be a strongly cyclic language, let $M$ be its syntactic monoid and let $P$
	be its syntactic image. By Proposition \ref{prop:identities of strongly cyclic},
	$P$ satisfies \ref{eq:S1} and \ref{eq:S2}. It suffices now to prove  that it 
	satisfies \ref{item:conjugation}. The sequence of implications
\begin{align*}
	xy \in P &\overset{\text{\ref{eq:S2}}}{\iff} (xy)^\omega \in P \iff (xy)^\omega
	(xy)^\omega \in P \iff (xy)^{\omega -1}xy(xy)^{\omega -1}xy \in P \\
	&\iff ((xy)^{\omega -1}x)(yx)^\omega y \in P \overset{\text{\ref{eq:S1}}}{\implies}
	(yx)^\omega \in P \overset{\text{\ref{eq:S2}}}{\iff} yx \in P.
\end{align*}
shows that $xy \in P$ implies $yx \in P$ and the opposite implication follows by
symmetry.
\end{proof}

\noindent Another result is worth mentioning: for any regular cyclic language, there
is a least strongly cyclic language containing it \cite[Theorem 2]{Carton97}.

\begin{prop}\label{prop:closure}
	Let $L$ be a regular cyclic language of $A^*$, let $\eta : A^* \to M$ be its
	syntactic stamp and let $P = \eta(L)$. There $M$ has a zero and the language 
\[
	\overline{L} = 
	\begin{cases} 
	  \eta^{-1}(\{s \mid s^\omega \neq 0\}) &\text{if $0 \notin P$,} \\
	  A^*    &\text{otherwise.}
	\end{cases}
\]
is the least strongly cyclic language containing~$L$.
\end{prop}

\begin{proof}
If $0 \notin P$, then the language $\overline{L}$ is strongly cyclic by
Proposition~\ref{prop:strongly cyclic languages}. Morevover, since $L$ is cyclic, $P$
is cyclic by Proposition \ref{prop:charcyclic}. It follows that if $s \in P$, then
$s^\omega \in P$ and in particular $s^\omega \neq 0$. Consequently, $\overline{L}$
contains $L$. 

It remains to prove that $\overline{L}$ is the least strongly cyclic language
containing $L$. Let $X$ be a strongly cyclic language containing~$L$ and let $u$ be a
word of $\overline{L}$. Let $\mathcal{A} = (Q,A,E)$ be a deterministic automaton such
that $X = \Stab(\mathcal{A})$. Setting $s = \eta(u)$, one has $s^\omega \neq 0$ by
definition of $\overline{L}$. Consequently, $\eta(s)^n \neq 0$ for every integer~$n$
and there are two words $x_n$ and~$y_n$ such that $x_nu^ny_n$ belongs to~$L$. By
Proposition \ref{prop:charstabword}, there is a state $q_n$ of~$\cA$ such that the
transition $q_n\cdot x_nu^ny_n$ is defined. The transition $(q_n\cdot x_n)\cdot u^n$
is thus defined for every $n$ and by Proposition \ref{prop:charstabword} again, the
word~$u$ belongs to~$X$. Thus $\overline{L} \subseteq X$ as required.
  
Suppose now that $0 \in P$ and let $z$ be a word of $L$ such that $\eta(z) = 0$. Let
$X$ be a strongly cyclic language containing $L$. If $X$ is not full, then $X$ is
nondense by Proposition \ref{prop:identities of strongly cyclic} and there exists a
word $u \in A^*$ such that $A^*uA^* \cap X = \emptyset$. Since $X$ contains $L$, one
also gets $A^*uA^* \cap L = \emptyset$ and in particular $zu \notin L$. But this
yieds a contradiction, since $\eta(zu) = \eta(z)\eta(u) = 0 \in P$ and thus $zu \in
\eta^{-1}(P) = L$. Thus the only strongly cyclic language containing~$L$ is $A^*$.
\end{proof}

\noindent Given a finite monoid $M$, the Green's preorder relation $\leJ$ defined on
$M$ by
\[
	s \leJ t \text{ if and only if $s \in MtM$, or equivalently, if there exists $u, v
	\in M$ such that $s = utv$}
\]
is a preorder on $M$. The associated equivalence relation $\cJ$ is defined by 
\[
	s \calJ t \text{ if $s \leJ t$ and $t \leJ s$, or equivalently, if $MsM = MtM$.}
\]

\begin{cor}\label{cor:strongly cyclic versus cyclic}
	Let $L$ be a regular cyclic language of $A^*$, let $\eta : A^* \to M$ be its
	syntactic stamp and let $P = \eta(L)$. Then $L$ is strongly cyclic if and only if
	for all idempotents $e, f$ of $M$, the conditions $e \in P$ and $e \leJ f$ imply $f
	\in P$.
\end{cor}

\begin{proof}
Suppose that $L$ is strongly cyclic and let $e, f$ be two idempotents of $M$ such
that $e \in P$ and $e \leJ f$. Let $u,v \in M$ be such that $e = ufv$. Since
$f^\omega = f$, one gets $uf^\omega v \in P$ and thus $f \in P$ by Condition
\ref{eq:S1} of Proposition \ref{prop:identities of strongly cyclic}.

In the opposite direction, suppose that for all idempotents $e, f$ of $M$, the
conditions $e \in P$ and $e \leJ f$ imply $f \in P$. Since $L$ is cyclic, it
satisfies \ref{item:power} and hence \ref{eq:S2}. We claim that it also satisfies
\ref{eq:S1}. Indeed, $ux^\omega v \in P$ implies $(ux^\omega v)^\omega \in P$ by
\ref{eq:S2}. Furthermore, since $(ux^\omega v)^\omega \leJ x^\omega$, one also has
$x^\omega \in P$, and finally $x \in P$ by \ref{eq:S2}, which proves the claim.
\end{proof}

The precise connection between cyclic and strongly cyclic languages was given in
\cite{BealCartonReutenauer96}.

\begin{thm}\label{thm:cyclic versus strongly cyclic}
	A regular language is cyclic if and only if it is a Boolean combination of regular
	strongly cyclic languages.
\end{thm}

\noindent Theorem \ref{thm:cyclic versus strongly cyclic} motivates a detailed
study of the difference hierarchy of the class $\cS$ of strongly cyclic languages.
This study relies on a careful analysis of the chains on the set of idempotents of a
finite monoid, pre-ordered by the relation $\leJ$.

\begin{defi}\label{def:P-chains of idempotents}
	A \emph{$P$-chain of idempotents} is a sequence $(e_0, e_1, \ldots, e_{m-1})$ of
	idempotents of $M$ such that
\[
  e_0 \leJ e_1 \leJ {}\dotsm{} \leJ e_{m-1}
\]
$e_0 \in P$ and, for $0 < i <m$, $e_i \in P$ if and only if $e_{i-1} \notin
P$. The integer $m$ is the length of the $P$-chain of idempotents.
\end{defi}

\noindent We let $\ell(M, P)$ denote the maximal length of a $P$-chain of idempotents
of $M$. We consider in particular the case where $\varphi : A^* \to M$ is a stamp
recognising a regular language $L$ of $A^*$ and $P = \varphi(L)$. The next theorem
shows that in this case, $\ell(M, P)$ does not depend on the choice of the stamp
recognising $L$, but only depends on $L$.


\begin{thm}\label{thm:P Q}
	Let $L$ be a regular language. Let $\varphi : A^* \to M$ and $\psi : A^* \to N$ be
	two stamps recognising $L$. If $P = \varphi(L)$ and $Q = \psi(L)$, then $\ell(M,P) =
	\ell(N,Q)$.
\end{thm}

\begin{proof}
It is sufficient to prove the result when $\varphi$ is the syntactic stamp of $L$.
Since the morphism~$\psi$ is surjective, $M$ is a quotient of $N$ and there is a
surjective morphism~$\pi : N \to M$ such that $\pi \circ \psi = \varphi$. It follows
that
\begin{equation}\label{eq:P and Q}
  \pi(Q) = P \text{ and } \pi^{-1}(P) = Q.
\end{equation}
We show that to any $P$-chain of idempotents in~$N$, one can associate a $Q$-chain of
idempotents of the same length in~$M$ and vice-versa.
  
Let $(e_0, \ldots, e_{m-1})$ be a $Q$-chain of idempotents in~$N$ and let $f_i =
\pi(e_i)$ for $0 \leq i \leq m-1$. Since every monoid morphism preserves $\leJ$, the
relations \eqref{eq:P and Q} show that $(f_0, \ldots, f_{m-1})$ is a $P$-chain of
idempotents in~$M$.
  
Let now $(f_0, \ldots, f_{m-1})$ be a $P$-chain of idempotents in~$M$. Since $f_{i-1}
\leJ f_i$, there exist for $1 \leq i \leq m-1$ elements $u_i, v_i$ of~$M$ such that
$u_if_iv_i = f_{i-1}$. Let us choose an idempotent $e_{m-1}$ such that $\pi(e_{m-1})
= f_{m-1}$ and some elements $s_i$ and $t_i$ of~$N$ such that $\pi(s_i) = u_i$ and
$\pi(t_i) = v_i$. We now define a sequence of idempotents $(e_0, \ldots, e_{m-1})$
of~$N$ by setting
\[
	e_{m-2} = (s_{m-1}e_{m-1}t_{m-1})^\omega \qquad e_{m-3} =
	(s_{m-2}e_{m-2}t_{m-2})^\omega \qquad \dotsm \qquad e_0 = (s_1e_1t_1)^\omega
\]
By construction, $e_0 \leJ {} \dotsm {} \leJ e_{m-1}$ and a straightforward induction
shows that $\pi(e_i) = f_i$ for $0 \leq i \leq m-1$. Moreover the equalities
\eqref{eq:P and Q} show that $e_i \in Q$ if and only if $f_i \in P$. It follows that
$(e_0, \ldots, e_{m-1})$ is a $Q$-chain of idempotents of $N$ and thus $\ell(M,P) =
\ell(N,Q)$.
\end{proof}
 
Since the integers $\ell(M, P)$ only depend on $L$ and not on the choice of the
recognising monoid, let us define $\ell(L)$ as $\ell(M, P)$ where $M$ $[P]$ is the
syntactic monoid [image] of $L$. Note that by Corollary \ref{cor:strongly cyclic
versus cyclic}, a cyclic language $L$ is strongly cyclic if and only if $\ell(L) =
1$. This is a special case of the following stronger result \cite[Theorem
4]{Carton97}.

\begin{thm}\label{thm:cyclic languages}
	Let $L$ be a regular cyclic language. Then $L \in \cB_n(\cS)$ if and only if
	$\ell(L) \leq n$.
\end{thm}

We first prove the following lemma which states that the function~$\ell$ is
subadditive with respect to the symmetric difference.

\begin{lem}\label{lem:subadditive}
	If $X$ and $Y$ are regular languages, then $\ell(X \dif Y) \leq \ell(X) + \ell(Y)$.
\end{lem} 

\begin{proof}
	Suppose that the languages $X$ and~$Y$ are respectively recognised by the stamps
	$\varphi : A^* \to M$ and $\psi : A^* \to N$. Let $P$ and~$Q$ be the images of $X$
	and~$Y$ in $M$ and~$N$, so that $X = \varphi^{-1}(P)$ and $Y = \psi^{-1}(Q)$. The
	language $X \dif Y$ is recognised by the restricted product of the stamps $\varphi$
	and $\psi$, say $\gamma: A^* \to R$, and the image of $X \dif Y$ in~$R$ is
\[
	T = {R \cap {\Bigl(P \times (N - Q) + (M - P) \times Q\Bigr)}}.
\]
	Let $((e_0,f_0), \ldots, (e_{m-1}, f_{m-1}))$ be a $T$-chain of idempotents in~$R$.
	Let us consider the set $I$ (resp. $J$) of integers~$i$ for which exactly one of
	the idempotents $e_{i-1}$ or $e_i$ (resp. $f_{i-1}$ or $f_i$) belongs to~$P$ (resp.
	$Q$). Formally, we define the sets of integers $I$ and~$J$ to be
  \begin{align*}
		 I &= \{ 1 \leq i \leq m-1 \mid e_{i-1} \in P \iff e_i \not\in P \} \\
		 J &= \{ 1 \leq i \leq m-1 \mid f_{i-1} \in Q \iff f_i \not\in Q \}
  \end{align*}
	Since the sequence $((e_0,f_0), \ldots, (e_{m-1}, f_{m-1}))$ is a $T$-chain in $R$,
	one has $e_0 \leJ {} \ldots {} \leJ e_{m-1}$ and $f_0 \leJ {} \ldots {} \leJ
	f_{m-1}$. Moreover, every integer $i$ between $1$ and $m-1$ belongs to exactly one
	of the sets $I$ or~$J$. Otherwise, the idempotents $(e_{i-1}, f_{i-1})$ and $(e_i,
	f_i)$ of~$R$ would be either both in $T$ or both out of $T$. Let $I = \{i_1,
	\ldots, i_p\}$ and $J = \{j_1, \ldots, j_q\}$ with $i_1 < {} \dotsm {} < i_p$ and
	$j_1 < {} \dotsm {} < j_q$. Then $p + q = m - 1$.
	
	Since $(e_0,f_0) \in T$, the conditions $e_0 \in P$ and $f_0 \notin Q$ are
	equivalent. By symmetry, suppose that $e_0 \in P$. Then $f_0 \notin Q$ and thus
	$f_1 \in Q$. Furthermore, the definitions of $I$ and $J$ give
\begin{align*}
	e_0 &\in P, & e_1 &\in P, &\ldots&& e_{i_1-1} &\in P, & e_{i_1} &\notin P, &\ldots&&
	e_{i_2-1} &\notin P, &e_{i_2} &\in P, &\ldots \\
	f_0 &\notin P, & f_1 &\notin P, &\ldots&& f_{j_1-1} &\notin P,& f_{j_1} &\in P,
	&\ldots&& f_{j_2-1} &\in P, &f_{j_2} &\notin P, &\ldots
\end{align*}
	Then the sequence $(e_0, e_{i_1}, \ldots, e_{i_p})$ is a $P$-chain of
	idempotents in $M$ and $(f_{j_1}, \ldots, f_q)$ is a $Q$-chain of idempotents in
	$N$. Therefore $p + 1 \leq \ell(X)$, $q \leq \ell(Y)$ and $m = p + 1 + q \leq
	\ell(X) + \ell(Y)$. Thus $\ell(X \dif Y) \leq \ell(X) + \ell(Y)$.
\end{proof}

\noindent We can now complete the proof of Theorem \ref{thm:cyclic languages}. 

\begin{proof}
Let $\eta:A^* \to M$ be the syntactic stamp of $L$ and let $P = \eta(L)$. Let 
also $E(M)$ be the set of idempotents of $M$.
If $L \in \cB_n(\cF)$, then $L = L_1 \dif {} \dotsm {} \dif L_n$ for some strongly cyclic
languages $L_i$. By Corollary~\ref{cor:strongly cyclic versus cyclic}, one has
$\ell(L_i) = 1$ for $1 \leq i \leq n$ and thus $\ell(L) \leq n$ by Lemma
\ref{lem:subadditive}.
  
Suppose now that $\ell(L) \leq n$. For each idempotent~$e$ of~$M$, let $\ell(e)$
denote the maximal length of a $P$-chain of idempotents ending with $e$. Then
$\ell(e) \leq \ell(L)$ by definition. For each $i > 0$, let
\[
	P_i = \{s \in M \mid \ell(s^\omega) \geq i\} \quad \text{and} \quad L_i =
	\eta^{-1}(P_i)
\]
Let $e, f \in E(M)$. Since every idempotent $e$ satisfies $e^\omega = e$, the
conditions $e \in P_i$ and $e \leJ f$ imply $f \in P_i$. It follows by
Corollary~\ref{cor:strongly cyclic versus cyclic} that the languages $L_i$ are
strongly cyclic. We claim that
\begin{equation}\label{eq:decomposition of P}
	P = P_1 - P_2 + P_3 - P_4\ \ldots\ \pm P_m  
\end{equation}
First observe that since $L$ is cyclic, an element~$s$ of~$M$ belongs to~$P$ if and
only if $s^\omega$ belongs to~$P$. Moreover, $s^\omega \in P$ if and only if
$\ell(s^\omega)$ is odd. Since $\ell(P) \leq n$, one has $\ell(s^\omega) \leq n$ for
every $s \in M$ and thus $P_{n+1} = \emptyset$. Formula \eqref{eq:decomposition of P}
follows, since for each $r \geq 0$,
\[
	\{s \in M \mid \ell(s^\omega) = r \} = P_r - P_{r+1}.
\]
Moreover, one gets from \eqref{eq:decomposition of P} the formula
\begin{equation}\label{eq:decomposition of L}
	L = L_1 - L_2 + L_3 \ \ldots\ \pm L_n  
\end{equation}
which completes the proof of the theorem.
\end{proof}

Theorem \ref{thm:cyclic languages} can be used to give an another proof of Theorem
\ref{thm:cyclic versus strongly cyclic}. To get this result, we must prove that any
cyclic language belongs to the class $\cB_n(\cS)$ for some integer~$n$. By Theorem
\ref{thm:cyclic languages}, it suffices to prove that the length of the $P$-chains of
idempotents in a monoid recognising~$L$ is bounded. This is a consequence of the
following proposition \cite[Proposition 5]{Carton97}.

\begin{prop}\label{prop:}
	Let $L$ be a regular cyclic language. Let $\varphi : A^* \to M$ be a stamp
	recognising $L$ and let $P = \varphi(L)$. Then the length of any $P$-chain of
	idempotents is bounded by the $\cJ$-depth of $M$.
\end{prop}

\begin{proof}
Let $(e_0, \ldots, e_{n-1})$ be a $P$-chain of idempotents in~$M$. Then by definition
\[
  e_0 \leJ {} \ldots {} \leJ e_{n-1}.
\]
Moreover, if $e_{i-1} \calJ e_i$, then by \cite[Proposition 1.12]{Pin86}, the
idempotents $e_{i-1}$ and $e_i$ are conjugate. That is, there exist two elements $x$
and~$y$ of~$M$ such that $xy = e_{k-1}$ and $yx = e_k$. Since $L$ is cyclic, $P$ is
also cyclic by Proposition~\ref{prop:charcyclic} and \ref{item:conjugation} implies
that $e_{i-1} \in P$ if and only if $e_i \in P$, which contradicts the definition of
a $P$-chain of idempotents. It follows that the sequence $(e_0, \ldots, e_{n-1})$ is a
strict $<_{\cJ}$-chain and hence its length is bounded by the $\cJ$-depth of $M$.
\end{proof}

\begin{exa} \label{exam:exam1}
	Let $L$ be the cyclic language $(b+aa)^* + (ab^*a)^* +a^* - b^* + 1$. Its syntactic
	monoid is the monoid with zero presented by the relations $bb = b$, $a^3 = a$, $baa
	= a^2b$, $a^2ba = ba$, $bab = 0$. Its transition table and its $\cJ$-class
	structure are represented below. The syntactic image of~$L$ is $P = \{1, a, a^2,
	aba, a^2b\}$ and $(aba, b, 1)$ is a maximal $P$-chain of idempotents.
	
\bigskip

\begin{center}
	\setlength{\extrarowheight}{2pt}\small
	\begin{tabular}{|l<{\ptvi}r|c|c|c|c|c|c|c|c|}
	\hline
	\multicolumn{2}{|c|}{}&1 &2 &3 &4 &5 &6 &7 &8 \\
	\hline
	$*$ &$1$   &1 &2 &3 &4 &5 &6 &7 &8 \\
			&$a$   &3 &4 &5 &2 &3 &8 &2 &6 \\
	$*$ &$b$   &7 &0 &8 &4 &4 &0 &7 &8 \\
  $*$ &$a^2$ &5 &2 &3 &4 &5 &6 &4 &8 \\
			&$ab$  &8 &4 &4 &0 &8 &8 &0 &0  \\
			&$ba$  &2 &0 &6 &2 &2 &0 &2 &6 \\
  $*$ &$a^2b$&4 &0 &8 &4 &4 &0 &4 &8 \\
  $*$ &$aba$ &6 &2 &2 &0 &6 &6 &0 &0 \\
  $*$ &$bab$ &0 &0 &0 &0 &0 &0 &0 &0 \\
	\hline
	\end{tabular}
	\qquad
  \begin{tabular}{c}
	\setlength{\extrarowheight}{6pt}
	\begin{tabularx}{24pt}{|@{}X@{}|}
		 \hline
		 \wstar{\ 1} \\
		 \hline
	\end{tabularx}\\
	$\tvi$\\
  \setlength{\extrarowheight}{6pt}
	\begin{tabularx}{24pt}{|@{}X@{}|}
		\hline
		\wstar{\ b} \\
		\hline
	\end{tabularx}
	\qquad
	\setlength{\extrarowheight}{6pt}
	\begin{tabularx}{48pt}{|@{}X@{}|@{}X@{}|}
		\hline
		\wstar{\ a^2} &\nstar{a} \\
		\hline
	\end{tabularx}\\
$\tvi$\\
	\setlength{\extrarowheight}{6pt}
	\begin{tabularx}{60pt}{|@{}X@{}|@{}X@{}|}
		\hline
		\wstar{a^2b} &\nstar{ba} \\
		\hline
		\nstar{ab} &\wstar{aba} \\
		\hline
	\end{tabularx}\\
$\tvi$\\
	\setlength{\extrarowheight}{6pt}
	\begin{tabularx}{30pt}{|@{}X@{}|}
		\hline
		\wstar{bab} \\
		\hline
	\end{tabularx}\\
	\end{tabular}
\end{center}
\end{exa}


\section{Conclusion}\label{sec:Conclusion}

Difference hierarchies of regular languages form an appealing measure of complexity.
They can be studied from the viewpoint of descriptive set theory and automata theory
\cite{GlasserSchmitzSelivanov16} or from an algebraic perspective, as presented in
this paper. It would be interesting to compare these two approaches.

The results proposed by Glasser, Schmitz and Selivanov
\cite{GlasserSchmitzSelivanov16}, together with our new result on group languages,
give hope that more decidability results might be obtained in a near future. In
particular, the recent progress on concatenation hierarchies \cite{Pin17b,
PlaceZeitoun16, PlaceZeitoun17}, might lead to new decidability results for the
difference hierarchies induced by the lower levels of the Straubing-Th\'erien
hierarchy.

Let us conclude with an open problem:

\begin{qu}\label{Q:BL decidable BLn undecidable}
	Does there exist a lattice of regular languages $\cL$ and an integer $n$ such that
	the membership problems for $\cL$ and for $\cB(\cL)$ are decidable, but is
	undecidable for $\cB_n(\cL)$?
\end{qu}

If the answer to Question \ref{Q:BL decidable BLn undecidable} is positive, a more 
precise question can be raised:

\begin{qu}\label{Q:BLn decidable BL(n+1) undecidable}
	For each integer $n$, does there exist a lattice of regular languages $\cL$ such that
	the membership problems for $\cL$, $\cB(\cL)$ and $\cB_n(\cL)$ are decidable, but
	the membership problem for $\cB_{n+1}(\cL)$ is undecidable?
\end{qu}

\section*{Acknowledgments}\label{sec:Acknowledgments}

The authors would like to thank the anonymous referees, whose suggestions strongly
improved the quality of this paper.


\bibliography{DifferenceHierarchy}  

\def\No{\kern-.25em\lower.2ex\hbox{\char'27}}%
\def\showlabel#1{}%
\begin{thebibliography}{10}

\bibitem{BealCartonReutenauer96}
{\sc M.-P. B{\'e}al, O.~Carton and C.~Reutenauer}, Cyclic languages and
  strongly cyclic languages, in {\em STACS '96}, pp.~49--59, {\em Lect. Notes
  in Comput. Sci.}\/ vol.~1046, 1996.\showlabel{BealCartonReutenauer96}

\bibitem{CaiHartmanisetal88}
{\sc J.-Y. Cai, T.~Gundermann, J.~Hartmanis, L.~A. Hemachandra, V.~Sewelson,
  K.~Wagner and G.~Wechsung}, The {B}oolean hierarchy. {I}. {S}tructural
  properties, {\em {SIAM} J. Comput.}\/ {\bf 17},6 (1988),
  1232--1252.\showlabel{CaiHartmanisetal88}

\bibitem{CaiHemachandra86}
{\sc J.-Y. Cai and L.~Hemachandra}, The {B}oolean hierarchy: hardware over
  {NP}, in {\em Structure in complexity theory ({B}erkeley, {C}alif., 1986)},
  pp.~105--124, {\em Lect. Notes in Comput. Sci.}\/ vol.~223, Springer, Berlin,
  1986.\showlabel{CaiHemachandra86}

\bibitem{Carton96}
{\sc O.~Carton}, Chain automata, {\em Theoret. Comput. Sci.}\/ {\bf 161}
  (1996), 191--203.\showlabel{Carton96}

\bibitem{Carton97}
{\sc O.~Carton}, A hierarchy of cyclic languages, {\em R.A.I.R.O.-Informatique
  {T}h{\'e}orique et {A}pplications}\/ {\bf 31},4 (1997),
  355--369.\showlabel{Carton97}

\bibitem{CartonPerrin94}
{\sc O.~Carton and D.~Perrin}, Chains and superchains in $\omega$-semigroups,
  in {\em Semigroups, Automata and Languages}, J.~Almeida, G.~Gomes and
  P.~Silva (eds.), pp.~17--28, World Scientific,
  1994.\showlabel{CartonPerrin94}

\bibitem{CartonPerrin97b}
{\sc O.~Carton and D.~Perrin}, Chains and superchains for {$\omega$}-rational
  sets, automata and semigroups, {\em Int. J. Alg. Comput.}\/ {\bf 7},7 (1997),
  673--695.\showlabel{CartonPerrin97b}

\bibitem{CartonPerrin97}
{\sc O.~Carton and D.~Perrin}, The {W}adge-{W}agner hierarchy of
  $\omega$-rational sets, in {\em Automata, Languages and Programming},
  P.~Degano, R.~Gorrieri and A.~Marchetti-Spaccamela (eds.), pp.~17--35, {\em
  Lect. Notes in Comput. Sci.}\/ vol.~1256, Springer-Verlag,
  1997.\showlabel{CartonPerrin97}

\bibitem{CartonPerrin99}
{\sc O.~Carton and D.~Perrin}, The {W}agner hierarchy of $\omega$-rational
  sets, {\em Int. J. Alg. Comput.}\/ {\bf 9},5 (1999),
  597--620.\showlabel{CartonPerrin99}

\bibitem{GlasserSchmitz01}
{\sc C.~Glasser and H.~Schmitz}, The {B}oolean structure of dot-depth one, {\em
  J. Autom. Lang. Comb.}\/ {\bf 6},4 (2001), 437--452.
\newblock 2nd Workshop on Descriptional Complexity of Automata, Grammars and
  Related Structures (London, ON, 2000).\showlabel{GlasserSchmitz01}

\bibitem{GlasserSchmitzSelivanov16}
{\sc C.~Glasser, H.~Schmitz and V.~Selivanov}, Efficient algorithms for
  membership in {B}oolean hierarchies of regular languages, {\em Theoret.
  Comput. Sci.}\/ {\bf 646} (2016),
  86--108.\showlabel{GlasserSchmitzSelivanov16}

\bibitem{Hausdorff14}
{\sc F.~Hausdorff}, {\em Grundz\"uge der Mengenlehre. Mit 53 Figuren im Text.},
  Leipzig: Veit \& Comp., 1914.\showlabel{Hausdorff14}

\bibitem{Hausdorff49}
{\sc F.~Hausdorff}, {\em Grundz\"uge der {M}engenlehre}, Chelsea Publishing
  Company, New York, N. Y., 1949.\showlabel{Hausdorff49}

\bibitem{Hausdorff57}
{\sc F.~Hausdorff}, {\em Set theory}, Chelsea Publishing Company, New York,
  1957.
\newblock Translated by John R. Aumann, et al.\showlabel{Hausdorff57}

\bibitem{KoblerSchoningWagner87}
{\sc J.~K\"obler, U.~Sch\"oning and K.~W. Wagner}, The difference and
  truth-table hierarchies for {NP}, {\em {RAIRO} Inform. Th\'eor. Appl.}\/ {\bf
  21},4 (1987), 419--435.\showlabel{KoblerSchoningWagner87}

\bibitem{Pin86}
{\sc J.-E. Pin}, {\em Varieties of Formal Languages}, North Oxford Academic,
  London and Plenum, New York, 1986.\showlabel{Pin86}

\bibitem{Pin94}
{\sc J.-E. Pin}, Polynomial closure of group languages and open sets of the
  Hall topology, in {\em 21th ICALP}, Berlin, 1994, pp.~424--435, {\em Lect.
  Notes in Comput. Sci.}\/ n\No 820, Springer.\showlabel{Pin94}

\bibitem{Pin95a}
{\sc J.-E. Pin}, Finite semigroups and recognizable languages : an
  introduction, in {\em NATO Advanced Study Institute {\it Semigroups, Formal
  Languages and Groups}}, J.~Fountain (ed.), pp.~1--32, Kluwer academic
  publishers, 1995.\showlabel{Pin95a}

\bibitem{Pin95c}
{\sc J.-E. Pin}, $PG=BG$, a success story, in {\em NATO Advanced Study
  Institute {\it Semigroups, Formal Languages and Groups}}, J.~Fountain (ed.),
  pp.~33--47, Kluwer academic publishers, 1995.\showlabel{Pin95c}

\bibitem{Pin95b}
{\sc J.-E. Pin}, A variety theorem without complementation, {\em Russian
  Mathematics (Iz. VUZ)}\/ {\bf 39} (1995), 80--90.\showlabel{Pin95b}

\bibitem{Pin17b}
{\sc J.-{\'E}. Pin}, The Dot-Depth Hierarchy, 45 Years Later, in {\em The Role
  of Theory in Computer Science - Essays Dedicated to Janusz Brzozowski},
  S.~Konstantinidis, N.~Moreira, R.~Reis and S.~Jeffrey (eds.), pp.~177--202,
  Word Scientific, 2017.\showlabel{Pin17b}

\bibitem{PinReutenauer91}
{\sc J.-E. Pin and C.~Reutenauer}, A conjecture on the {H}all topology for the
  free group, {\em Bull. London Math. Soc.}\/ {\bf 23} (1991),
  356--362.\showlabel{PinReutenauer91}

\bibitem{PlaceZeitoun16}
{\sc T.~Place and M.~Zeitoun}, Separating Regular Languages with First-Order
  Logic, {\em Logical Methods in Computer Science}\/ {\bf 12},5 (2016),
  1--30.\showlabel{PlaceZeitoun16}

\bibitem{PlaceZeitoun17}
{\sc T.~Place and M.~Zeitoun}, Concatenation Hierarchies: New Bottle, Old Wine,
  in {\em Computer Science -- Theory and Applications: 12th International
  Computer Science Symposium in Russia, CSR 2017, Kazan, Russia, June 8-12,
  2017, Proceedings}, P.~Weil (ed.), pp.~25--37, Springer,
  2017.\showlabel{PlaceZeitoun17}

\bibitem{RibesZalesskii93}
{\sc L.~Ribes and P.~A. Zalesskii}, On the profinite topology on a free group,
  {\em Bull. London Math. Soc.}\/ {\bf 25},1 (1993),
  37--43.\showlabel{RibesZalesskii93}

\bibitem{Schutzenberger56}
{\sc M.~P. Sch{\"u}tzenberger}, Une th\'eorie alg\'ebrique du codage, {\em
  S\'eminaire Dubreil. Alg\`{e}bre et th\'eorie des nombres}\/ {\bf 9}
  (1955-1956), 1--24.
\newblock
  {\footnotesize\url{http://eudml.org/doc/111094}}.\showlabel{Schutzenberger56}

\bibitem{Simon75}
{\sc I.~Simon}, Piecewise testable events, in {\em Proc. 2nd GI Conf.},
  H.~Brackage (ed.), pp.~214--222, {\em Lecture Notes in Comp. Sci.}\/ vol.~33,
  Springer Verlag, Berlin, Heidelberg, New York, 1975.\showlabel{Simon75}

\bibitem{Wadge12}
{\sc W.~W. Wadge}, Early investigations of the degrees of {B}orel sets, in {\em
  Wadge degrees and projective ordinals. {T}he {C}abal {S}eminar. {V}olume
  {II}}, pp.~166--195, {\em Lect. Notes Log.}\/ vol.~37, Assoc. Symbol. Logic,
  La Jolla, {CA}, 2012.\showlabel{Wadge12}

\bibitem{Wagner79}
{\sc K.~Wagner}, On {$\omega $}-regular sets, {\em Inform. and Control}\/ {\bf
  43},2 (1979), 123--177.\showlabel{Wagner79}

\bibitem{Wechsung85}
{\sc G.~Wechsung}, On the {B}oolean closure of {NP}, in {\em Fundamentals of
  computation theory ({C}ottbus, 1985)}, pp.~485--493, {\em Lect. Notes in
  Comput. Sci.}\/ vol.~199, Springer, Berlin, 1985.\showlabel{Wechsung85}

\end{thebibliography}
\bibliographystyle{biblistanglais}

\end{document}